\documentclass{article}
\usepackage[a4paper, total={7in, 9in}]{geometry}

\usepackage{amsmath,amssymb,amsthm,mathtools,bm,booktabs,graphicx,array,multirow}
\usepackage{natbib}
\usepackage{xcolor}
\usepackage[hidelinks]{hyperref}
\usepackage{url}

\newtheorem{theorem}{Theorem}
\newtheorem{proposition}[theorem]{Proposition}
\newtheorem{corollary}[theorem]{Corollary}
\newtheorem{lemma}[theorem]{Lemma}
\theoremstyle{remark}
\newtheorem*{remark}{Remark}
\newcommand{\T}{\mathbb T}
\newcommand{\R}{\mathbb R}
\newcommand{\N}{\mathcal N}
\newcommand{\dd}{\,\mathrm d}
\newcommand{\MMD}{\mathrm{MMD}}
\newcommand{\dist}{\operatorname{dist}}
\newcommand{\E}{\mathbb E}

\newcommand{\argmin}{\operatorname*{arg\,min}}
\DeclareMathOperator{\pr}{pr}
\newcommand{\footremember}[2]{%
    \footnote{#2}
    \newcounter{#1}
    \setcounter{#1}{\value{footnote}}%
}
\newcommand{\footrecall}[1]{%
    \footnotemark[\value{#1}]%
} 
\title{How to bowl properly}
\author{%
  The Dude\footremember{alley}{Holly Star Lanes Bowling Alley}%
  \and Walter Sobchak\footremember{trailer}{Probably in a trailer park}%
  \and Jesus Quintana\footrecall{alley} \footnote{Mexico?}%
  \and Uli Kunkel\footrecall{trailer} \footnote{Germany?}%
  }

\begin{document}

%

%

\title{Projected Normal Models on Recursively Curved Tori}

\author{ Surojit Biswas \footremember{Bonn}{Hausdorff Research Institute for Mathematics, University of Bonn} \and Aron Persson \footrecall{Bonn} \footremember{Ume}{Department of Mathematics and Mathematical Statistics, Umeå University}}

\date{}

\maketitle

\begin{abstract}
We develop a  computationally tractable statistical framework for angular data supported on a recursively embedded curved torus $\mathbb{T}^n\hookrightarrow\mathbb{R}^{n+1}$.   Due to the recursive structure we build an exact uniform sampler by generalizing the process from the curved 2-torus.  We generalize projected normal distributions to the curved $n$-torus and we develop inference theory for these distributions. We utilize both maximum-log-likelihood estimates and energy-based Fr\'echet estimates for the projected normal distribution for this purpose. This is done by establishing identifiability for isotropic projected normal distributions and by applying theory of maximum likelihood and intrinsic statistics for the curved tori. We numerically verify  the geometry, sampler, density transformation, and energy solver. Lastly, we demonstrate the applicability of the proposed model using real data.
\end{abstract}

\section{INTRODUCTION}

Toroidal observations arise whenever several angular variables are
recorded jointly, including robot configurations (see \cite{mescher2024geometric,sato2023online}), ophthalmology \cite[see][]{BanerjeeBiswas2026}, protein backbone
dihedrals (see \cite{gottard2021graphical,hong2022clustorus}), genetic information (see \cite{Eltzner2018torus}), periodic signals (see \cite{klein2020torus}), and directional time series (see \cite{giannakis2019data}).  The usual
statistical representation is the flat product torus
$(\mathbb{S}^1)^n$, where dependence is introduced directly through
an angular density.

An embedded torus, however, may carry the non-flat Riemannian geometry
induced by its Euclidean embedding.  Its volume element, geodesic
distance, and normal directions therefore depend on the embedding and
not only on the product topology.  For the classical embedded curved
$2-$torus, \cite{biswas2026efficient} extended the von Mises distribution on curved $\mathbb{T}^2$,  \citet{Diaconis2013} showed that uniformity with respect to
surface area induces a non-uniform angular law.
\citet{BanerjeeBiswas2026} used the same curved $\mathbb{T}^2$
geometry to construct dependent toroidal distributions and an exact
cardioid transformation sampler.  These results motivate the use of
embedded geometry in toroidal data, extending the area law
and Cardioid sampler. It is a simple geometric extension and is not the main novelty of the present paper.

\paragraph{RELATED WORK.}
Directional-statistics models $\T^n$ include products of circular
laws, multivariate von Mises families, and wrapped Euclidean
distributions \cite[see][]{MardiaJupp2000,Mardia2008,mardia2014some,Greco2023,xu2023density}.  The sine model due to 
\cite{Singh2002Probabilistic}  represents pairwise dependence
through sine interactions, whereas multivariate wrapped normals used in \cite{nodehi2021estimation,nodehi2025torus} inherit
a Euclidean Gaussian dependence structure after summing over wrapping
vectors.  These models provide natural flat-torus comparators because
they model the same angular observations without using an embedded
curved geometry. 

Projected normal distributions provide a flexible route to circular
and directional modeling by mapping normal random variables onto a subspace, although standard parametrizations generally exhibit scale
non-identifiability \cite[see][]{WangGelfand2013}.  Multivariate projected
normal constructions for several circular variables are already established in the literature \cite[see][]{nodehi2025torus}.
In particular, \citet{Mastrantonio2018} constructs a multivariate
projected normal by partitioning a $2n$-dimensional Gaussian vector into
$n$ bivariate components and radially projecting each component.
More recently, \citet{KatoMastrantonioIshikawa2025} proposed an
interpretable toroidal projected-normal family for general-dimensional
hypertoroidal data together with a related copula construction.


\paragraph{Gap and positioning.}
The above strands address complementary aspects of toroidal modeling,
but they do not provide a single construction in which the probability
model, density calculation, and intrinsic statistics are all tied to the
geometry of one recursively embedded curved torus.  Existing
multivariate projected-normal models use coordinate-wise radial
projections onto a product of circles, whereas the curved-torus work
above is restricted to the ordinary $2-$torus and does not provide the
normal-fiber projected-Gaussian construction developed here.
Our objective is therefore not to introduce the first projected-normal
model on $\T^n$, nor merely to extend the surface-area sampler to higher
dimension.  Instead, we couple an ambient Gaussian model to the
normal-fiber geometry of a single recursively curved embedding
$
F_n(\mathbb{T}^n)\subset\mathbb{R}^{n+1},
$
so that the same recursive structure determines the metric, volume
element, normal field, fiber Jacobian, projected density, and intrinsic
geometric summaries.

The paper \textit{makes five contributions}.  \textit{First}, we organize the recursively
curved torus geometry into a form suitable for statistical modeling;
the metric is diagonal, the volume is explicit, and the normal-fiber
map is affine in its radial coordinate. \textit{ Second}, we define a
normal-fiber projected Gaussian family and derive its angular and
surface densities.  \textit{Third}, we establish global identifiability and likelihood-based
asymptotic theory for an isotropic subfamily.  For the more flexible family
used in the applications, we establish consistency of an idealized
simulation-based minimum-MMD estimator using a characteristic kernel
\cite[see][]{Gretton2012}.  \textit{Fourth}, we define intrinsic distance and
Fr\'echet statistics through the path-energy functional of the curved
metric, with minimization over winding classes.  \textit{Fifth}, we validate the construction numerically and
compare it with literature-established toroidal models on robot joint
trajectories and protein dihedral data.

\section{RECURSIVELY CURVED TORUS} 
\vspace{-.71em}
 Throughout, $\T^n=\R^n/(2\pi\mathbb{Z})^n$ denotes the flat torus equipped with the classical quotient topology. Moreover, we denote points as {$\bm{\theta}$}$=(\theta_1,\ldots,\theta_n)\in\T^n$ and we shall throughout assume that the radii $R_1,\ldots,R_n>0$ satisfy the strong radius condition
\begin{equation}
R_k>\sum_{j<k}R_j,\quad k=2,\ldots,n.
\label{eq:radius}
\end{equation}
As is done in \cite{llibre2024darboux}, we define recursively
\begin{equation}
\rho_1=R_1,
\quad
\rho_k=R_k+\rho_{k-1}\cos\theta_{k-1},
\label{eq:rho}
\end{equation}
and the geometry-defining parameterization\\ $F_n:\T^n\rightarrow \R^{n+1}$ is defined by
\begin{equation}
F_n({ \bm{\theta}})=(\rho_n\cos\theta_n,\rho_n\sin\theta_n,\rho_{n-1}\sin\theta_{n-1},\ldots,\rho_1\sin\theta_1).
\label{eq:embed}
\end{equation}
Condition \eqref{eq:radius} guarantees $\rho_k>0$ for every angular configuration and yields a smooth embedded hypersurface; $F_n(\T^n)\subset\R^{n+1}$ hence, from this point onward, we refer to it as the \textit{curved $n$-torus.}

\paragraph{Dependence on coordinate ordering.}
The recursive construction is intentionally ordered.  In general, a
permutation of the angular coordinates changes the intermediate radii
$\rho_k$ and hence changes the embedding, induced metric, and resulting
projected distribution.  Thus the proposed model is not permutation
invariant.  When the coordinates possess a natural sequential
interpretation, such as revolute joints along a serial kinematic chain,
we use that physical ordering a priori.  For applications in which no
scientifically distinguished ordering is available, the coordinate
order should instead be regarded as part of the model specification
and assessed by sensitivity analysis or selected using training or
validation data only.  A validation-only ordering sensitivity is reported in
Appendix~\ref{app:ordering-sensitivity}.

\begin{proposition}
\label{prop:geometry}
The following holds for the isometrically embedded Riemannian manifold $F_n( \bm{\theta})\subset \R^{n+1}$.  The induced metric is diagonal in $\bm{\theta}$-coordinates,
\begin{equation}
{g(\bm{\theta})}=\operatorname{diag}(\rho_1^2,\ldots,\rho_n^2),
\label{eq:riemannian_metric_local}
\end{equation}
and therefore
{
\begin{equation}
\operatorname{dVol}_n(\bm{\theta}) = \Bigl(\prod_{k=1}^n\rho_k\Bigr)\dd\bm{\theta}\
\label{eq:volume_form}
\end{equation}
and
\begin{equation}
\operatorname{Vol}(F_n(\T^n))=(2\pi)^n\prod_{k=1}^nR_k.
\label{eq:volume}
\end{equation}}
The  unit normal vector field can be written component-wise as
\begin{equation}
\begin{split}
\nu_1&=\prod_{j=1}^n\cos\theta_j,\\
\nu_{n-k+2}&=\Bigl(\prod_{j=1}^{k-1}\cos\theta_j\Bigr)\sin\theta_k,
\quad k=2,\ldots,n,\\
\nu_{n+1}&=\sin\theta_1.
\end{split}
\label{eq:normal}
\end{equation}

\end{proposition}

The proof is by direct differentiation and induction over the coordinates. See Figure \ref{fig:torus_with_normal} for a schematic of the embedding given by \eqref{eq:embed} with the normal field as in \eqref{eq:normal}.

\begin{figure}[h!]
\centering
\includegraphics[width=0.42\textwidth]{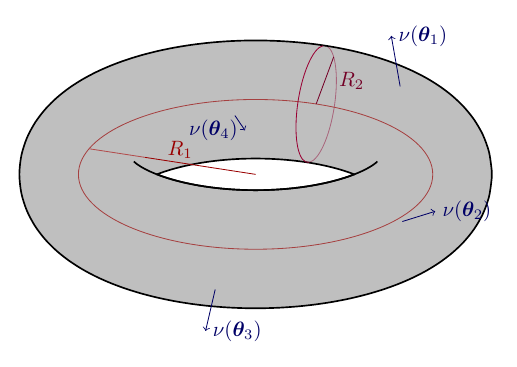}
\caption{A schematic of the embedded 2-torus with radii $R_1$ and $R_2$ and with normal vector field $\nu$ at points represented by angles $\bm{\theta}_1,\bm{\theta}_2,\bm{\theta}_3, \bm{\theta}_4\in \mathbb{T}^2$.}
\label{fig:torus_with_normal}
\end{figure}

\subsection{Intrinsic uniform law and exact sampling}
Following the work by \citet{Diaconis2013}, normalizing {the volume form \eqref{eq:volume_form} by the total volume \eqref{eq:volume}} gives the angular density of the intrinsic uniform distribution,
\begin{equation}
p_U({ \bm{\theta}})=\frac{1}{(2\pi)^n}
\prod_{j=1}^{n-1}
\left(1+\frac{\rho_j}{R_{j+1}}\cos\theta_j\right).
\label{eq:uniform}
\end{equation}
Because $\rho_j$ depends only on $\theta_1,\ldots,\theta_{j-1}$,  { the marginal distribution may be written as a product of conditional distributions.}  Specifically, 
\[
\Theta_j\mid\Theta_{1:j-1}
\sim\frac{1}{2\pi}(1+a_j\cos\theta_j),
\quad a_j=\rho_j/R_{j+1},
\]
for $j<n$, while $\Theta_n\sim U(0,2\pi)$.  { By assumption} \eqref{eq:radius} $0<a_j<1$.  Thus the exact cardioid transformation of \citet{BanerjeeBiswas2026} can be applied sequentially, giving an exact  {$\mathcal{O}(n)$} surface-uniform sampler.

\begin{theorem}
\label{thm:exactuniform_alg}
Initialize $\rho\leftarrow R_1$.  For $j=1,\ldots,n-1$, set
$a=\rho/R_{j+1}$, draw $X\sim U(0,2\pi)$ and $U\sim U(0,1)$, and let
\[
p(X)=\frac{1+a\cos X}{2}.
\]
Set
\[
\Theta_j=
\begin{cases}
X,&U<p(X),\\
\pi-X,&U\ge p(X),\ X<\pi,\\
3\pi-X,&U\ge p(X),\ X>\pi,
\end{cases}
\]
then update $\rho\leftarrow R_{j+1}+\rho\cos\Theta_j$.  Finally draw
$\Theta_n\sim U(0,2\pi)$.  Then $F_n(\Theta)$ is exactly uniform with respect to normalized Riemannian volume.
\end{theorem}
 See Appendix \ref{sec:proof_exactuniform_alg} for a proof that this algorithm is exact. 

\section{NORMAL-FIBER PROJECTED GAUSSIAN}
In this section we shall describe the geometry and the uniform distribution of the curved $n-$torus $F_n(\T^n)$.

\subsection{Normal-fiber coordinates}
{In order to define the (maximal) tubular neighborhood around the torus $F_n(\mathbb{T}^n) 
\subset\mathbb{R}^{n+1}$, we shall replace the innermost radius $R_1$ by the fiber coordinate $T>0$ and define} 
\begin{equation}
q_1(T,{\bm{\theta}})=T,
\quad
q_k(T,{\bm{\theta}})=R_k+q_{k-1}(T,{\bm{\theta}})\cos\theta_{k-1}.
\label{eq:q}
\end{equation}
Now an ambient point can be expressed in the coordinates $(T,\bm{\theta}) $ by replacing every $\rho_k$ by $q_k$ in \eqref{eq:embed}. That is we define the normal-fiber coordinate function by
\begin{equation}
X(T,\bm{\theta}) =(q_n\cos\theta_n,q_n\sin\theta_n,
q_{n-1}\sin\theta_{n-1}, \ldots, q_1\sin\theta_1).
    \label{eq:fiber-map}
\end{equation}
 which  satisfies the affine identity
\begin{equation}
X(T,{\bm{\theta}})=F_n({\bm{\theta}})+(T-R_1)\nu({\bm{\theta}}).
\label{eq:fiber-affine}
\end{equation}
The  metric { of the tubular neighborhood  in $(T,\bm{\theta})$ coordinates } is
\begin{equation}
\dd s^2=\dd T^2+\sum_{k=1}^nq_k(T,{\bm{\theta}})^2\dd\theta_k^2,
\label{eq:fiber-metric}
\end{equation}
and the  Jacobian { determinant of the tubular neighborhood parametrization $X$} is
\begin{equation}
J(T,{\bm{\theta}})=T\prod_{k=2}^nq_k(T,{\bm{\theta}}).
\label{eq:fiber-jac}
\end{equation}
 { Note that every $q_k$ is affine in $T$ and that $q_k(T,\bm{\theta})>0$ for all $0<T\leq R_1$ by \eqref{eq:radius}. In order for the derivative of the normal fiber coordinates to attain maximal rank, \eqref{eq:fiber-jac} implies that each $q_k$ remains positive. Therefore, we define, for fixed $\bm{\theta}= ( \theta_1,\dots,\theta_n)$, the upper bound for $T$ by
\begin{equation}
T_0(\bm{\theta})=\inf_{k=1,\dots,n} \{T>0: q_k(T,\bm{\theta})= 0 \}
\label{eq:Tzero}
\end{equation}
and naturally $T_0(\bm{\theta})= \infty$ in case the set 
\[
\bigcup_{k=1}^n\{T>0: q_k(T,\bm{\theta})= 0 \}
\]
 is empty. Hence the normal fiber coordinates $X$ is defined for all $(T,\bm{\theta})$ such that $0<T< T_0(\bm{\theta})$ and $\bm{\theta}\in  \mathbb{T}^n$. On this domain, denoted $\mathcal{D}$, $X$ is one-to-one, and its inverse is obtained recursively by two-dimensional polar inversions. We obtain the projection map $\pr:X(\mathcal{D})\rightarrow F_n(\mathbb{T}^n)$ defined by $\pr(x_0,x_1,\dots,x_n)=F_n(X^{-1}(x_0,x_1,\dots,x_n))$ where $F_n$ is extended to the coordinates $(T,\bm{\theta})$ by defining $F_n(T,\bm{\theta}):= F_n(\bm{\theta}).$ See Appendix~\ref{sec:normalfibercoordinatelemma} for more details.
}

\subsection{Projected Gaussian density}
Let $Z\sim\N_{n+1}(\mu,\Sigma)$, and define { the projected normal as the random variable $\pr(Z)$ which is almost surely defined if $\mathbb{P}(Z\in X(\mathcal{D}))=1$.}  
{ We define $a(\bm{\theta}):= F_n({\bm{\theta}})-R_1\nu(\bm{\theta})-\mu$ and denote
\begin{align*}
A(\bm{\theta})&=\nu({\bm{\theta}})^\top \Sigma^{-1}\nu({\bm{\theta}}),&
B({\bm{\theta}})&=\nu({\bm{\theta}})^\top \Sigma^{-1}a(\bm{\theta}),\\
C(\bm{\theta})&=-\frac{1}{2}a(\bm{\theta})^\top \Sigma^{-1}a(\bm{\theta}),& K(\bm{\theta}) &= \frac{B(\bm{\theta})}{\sqrt{A(\bm{\theta})}}.
\end{align*}
}

\begin{theorem}
\label{thm:pgdensity} The probability density function of $\pr(Z)$ in the angular coordinates $\bm{\theta}\in \T^n$  is given by
\begin{multline}
p_{\operatorname{pr}(Z)}(\bm{\theta})=\frac{\exp\big(C(\bm{\theta}) + \frac{1}{2} K^2(\bm{\theta})\big)}
{(2\pi)^{(n+1)/2}|\Sigma|^{1/2} A(\bm{\theta})}
\int_{K(\bm{\theta})}^{\sqrt{A}(\bm{\theta})T_0(\bm{\theta})+K(\bm{\theta})} 
\exp\left(-\frac{1}{2}S^2\right)
\\(S-K(\bm{\theta}))  \prod_{k=2}^n q_k\left(\frac{S-K(\bm{\theta})}{\sqrt{A}(\bm{\theta})},\bm{\theta}\right)\,\dd S.
\label{eq:angular-density}
\end{multline}

{The corresponding push forward probability measure on the embedded torus satisfies
\begin{equation}
\mathbb{P}\big(\pr(Z)\in F_n(V)\subseteq F_n(\mathbb{T}^n)\big)=\int_{\bm{\theta}\in V} \frac{p_{\pr(Z)}(\bm{\theta})}{J(R_1,\bm{\theta})} \operatorname{dVol}(\bm{\theta}).
\label{eq:surface-density}
\end{equation}}
\end{theorem}
Theorem~\ref{thm:pgdensity} follows directly from Gaussian change of variables in $(T,\theta)$ followed by integration along each normal fiber, see Appendix~\ref{sec:proof_prob_density}. 

Since
\[
T\prod_{k=2}^nq_k(T,{\bm{\theta}})=\sum_{m=1}^{n}c_m({\bm{\theta}})T^m,
\]
the integral in \eqref{eq:angular-density} is a finite linear combination of truncated Gaussian moments. How these moments can be computed is given explicitly in Appendix~\ref{sec:truncated_form_pg}.

{
\subsection{Identifiability of an isotropic projected Gaussian}
Standard projected-normal models generally require identifying
constraints because of scale non-identifiability
\cite[see][]{WangGelfand2013}.  For the recursively curved embedding studied
here, however, this scale ambiguity is removed in the isotropic
subfamily. Theorem~\ref{thm:ident} establishes global identifiability
of $(\mu,\sigma)$ for $n\ge2$.

\begin{theorem}
\label{thm:ident}
Let $Z \sim N_{n+1}(\mu,\sigma^2 I)$ and let $Z'\sim N_{n+1}(\tilde{\mu},\tilde{\sigma}^2 I)$ with $\sigma,\tilde{\sigma}>0$. Then it holds that $\pr(Z)\sim \pr(Z')$ if, and only if, $\mu=\tilde{\mu}$ and $\sigma= \tilde{\sigma}$.
\end{theorem}

The idea behind the proof uses that the projection preserves the signature of the first two coordinates of a random variable $Z$ which allows us to identify standardized coordinates $\frac{\mu_0}{\sigma},\frac{\mu_1}{\sigma}$. Then, the parameter $\sigma$ is identified by the non-central chi distribution of $\sqrt{Z_0^2+Z_1^2}$ which appears in the third coordinate. Lastly, the reductive curvature structure of the torus gives the rest of the parameters $\mu_2,\dots,\mu_n$, see Appendix~\ref{sec:proof_indentifiability} for full proof.
}

\section{STATISTICAL INFERENCE}
\subsection{Maximum likelihood theory}
For the identifiable isotropic family let $\gamma=(\mu^\top,\sigma)^\top$ and let $\ell_N(\gamma)$ be the log-likelihood { of the probability density function given in} \eqref{eq:angular-density}.  Furthermore, we consider the compact parameter set such that $F_n(\mathbb{T}^n)$ is contained in the interior for the mean parameter space, then differentiability and domination of the density, and positive-definite Fisher information, standard M-estimation arguments apply because the torus is compact and the density is smooth away from measure-zero coordinate boundaries, see Appendix~\ref{sec:basics_max_log_likelihood} for an overview of the definitions used in the maximum log-likelihood estimation.

\begin{theorem}
\label{thm:mle}
{Let $\Theta_1,\dots,\Theta_N$ be independent identically distributed observations following the distribution of $\pr(Z)$ with $Z\sim N_{n+1}(\mu_0,\sigma_0^2 I)$ and} {assume the parameter space $\Xi$ in} \eqref{eq:compact-parameter-space} { is compact with $\vartheta_0=(\mu_0^T,\sigma_0)^T\in \operatorname{int}\Xi$}. Then any maximizer
$\widehat\vartheta_N$ { of the log-likelihood defined by } \eqref{eq:mle-def} satisfies
\begin{equation}
\widehat\vartheta_N\longrightarrow\vartheta_0
\quad\text{almost surely}.
\label{eq:mle-strong-consistency}
\end{equation}
If, in addition, the Fisher information matrix $\mathcal I(\vartheta_0)$ is positive definite, then
\begin{equation}
\sqrt N\big(\widehat\vartheta_N-\vartheta_0\big)
\ \xrightarrow{{\dd}}\ 
\mathcal N_{n+2}
\left(0,\mathcal I(\vartheta_0)^{-1}\right).
\label{eq:mle-clt}
\end{equation}
\end{theorem}
See Appendix~\ref{sec:proof_mle} for the proof.

\subsection{Minimum-MMD estimation for the flexible family}

For applications, the recursive geometry is fixed and we use the
structured ambient covariance
\begin{equation}
\Sigma=
\operatorname{diag}
\left(
\lambda_{\min}+e^{2s_1},\ldots,
\lambda_{\min}+e^{2s_{n+1}}
\right)
+vv^\top,
\label{eq:flexcov}
\end{equation}
where $\lambda_{\min}>0$ is a fixed covariance floor,
$s,v\in\mathbb{R}^{n+1}$, and
$\eta=(\mu,s,v)$ denotes the full ambient-Gaussian parameter.
The diagonal component controls coordinate-specific variation, while
$vv^\top$ provides a parsimonious dependence structure.

Let $P_\eta$ denote the angular distribution obtained by projecting
$Z\sim N_{n+1}(\mu,\Sigma)$ onto the curved $n-$torus $F_n(\T^n)$.
We estimate $\eta$ by minimum MMD after the sine-cosine embedding
$$
\Phi(\bm{\theta})
=
(\cos\theta_1,\ldots,\cos\theta_n,
 \sin\theta_1,\ldots,\sin\theta_n).
$$
Let $\widehat J_{N,m}(\eta)$ be the simulated empirical squared-MMD
criterion based on $N$ observations and $m$ model-generated draws.

\begin{proposition}
\label{prop:mmd}
Assume the parameter space $\mathcal H$ is compact,
$\eta\mapsto P_\eta$ is continuous in the MMD metric, the kernel is
bounded and characteristic, $m=m_N\to\infty$, and
\begin{align*}
\sup_{\eta\in\mathcal H}
\left|
\widehat J_{N,m_N}(\eta)-J(\eta)
\right|
&\to_p0,\\
J(\eta)
&=\MMD^2(P_0,P_\eta).
\end{align*}
If $\widehat\eta_N$ is an $o_p(1)$-approximate minimizer, then
$$
\dist\!\left(
\widehat\eta_N,
\argmin_{\eta\in\mathcal H}J(\eta)
\right)\to_p0.
$$
If $P_0=P_{\eta_0}$, then
$
\MMD(P_{\widehat\eta_N},P_{\eta_0})\to_p0;
$
parameter consistency follows whenever the chosen parametrization is
identifiable.
\end{proposition}

Proposition~\ref{prop:mmd}
concerns an idealized asymptotic minimum-MMD estimator with increasing
simulation size and asymptotically negligible optimization error. The proof and implementation
details are given in Appendix~\ref{app:flexible-mmd}.

\subsection{Intrinsic energy and Fr\'echet summaries}

For a smooth curve $\gamma:[0,1]\rightarrow \mathbb{R}^n$, define the energy functional as 
 \begin{equation}
\mathcal{E}[\gamma]=\int_0^1
\sum_{j=1}^n\rho_j(\gamma(t))^2\dot\gamma_j(t)^2\,\dd t.
\label{eq:energy}
\end{equation}
This is the usual Riemannian energy functional
\cite[see][Section~3.2]{srivastava2016functional}, expressed in the
angular coordinates of $F_n(\T^n)$.  The following proposition shows
that the intrinsic geodesic distance on the embedded $n$-torus, $F_n(\mathbb{T}^n)$ can be
obtained by minimizing this energy over all winding classes of the
universal cover $\mathbb{R}^n$ of $\T^n$.

\begin{proposition}
\label{prop:winding-attainment}

For points on $F_n(\T^n)$ represented by angular coordinates
$\bm{\alpha},\bm{\beta}\in\T^n$, the squared intrinsic geodesic
distance is
\begin{equation}
d_M^2(\bm{\alpha},\bm{\beta})=
\min_{k\in\mathbb Z^n}
\inf_{\substack{\gamma(0)=\bm{\alpha}\\
\gamma(1)=\bm{\beta}+2\pi k}}\mathcal{E}[\gamma].
\label{eq:distance}
\end{equation}
where $d_M$ is the Riemannian geodesic distance induced by the Riemannian metric
\eqref{eq:riemannian_metric_local}.
\end{proposition}

The population Fr\'echet mean is defined as a minimizer of
$$
q\longmapsto \E\!\left[d_M^2(q,Y)\right].
$$
Since $F_n(\T^n)$ is compact, a population minimizer exists.  If this
minimizer is unique, the empirical Fr\'echet mean is 
consistent.  The intrinsic covariance is defined through the initial
velocities of the corresponding minimizing energy paths.  A sufficient
strong-convexity condition for uniqueness, together with the associated
consistency result for the intrinsic covariance, is given in
Appendix~\ref{sec:energy-frechet}; see also
\citet{Afsari2011} and \citet{bhattacharya2003large}.
\section{NUMERICAL VALIDATION}
\label{sec:numerical-validation}
We use $R_1=1$ and $R_k=1.05\sum_{j<k}R_j$ unless otherwise stated. We validate each of the following separately, geometry and uniform sampling, inverse normal-fiber coordinates, density normalization and Monte Carlo agreement, restricted-parameter recovery, intrinsic-energy optimization, and runtime. Complete numerical settings and tables are reported in
Appendix~\ref{app:simulation-validation}.

\subsection{Geometry, sampler, and density checks}
For dimensions $n\in\{2,3,5,10\}$, the numerical validation, summarized in
Table~\ref{tab:audit}, shows that Monte Carlo estimates of the
Riemannian volume agree with $(2\pi)^n\prod_k R_k$ to within
$2.1\%$, with errors below $0.5\%$ through $n=5$.  Across $20{,}000$ exact sampler draws, the largest probability-integral-transform KS discrepancy is $0.0115$ and the largest absolute transformed-coordinate correlation is below $0.02$.  Inverse normal-fiber coordinates reconstruct ambient draws to floating-point accuracy through $n=10$.

\begin{table}[h!]
\vspace{-1em}
\centering
\caption{Selected geometry and computational validations.}
\label{tab:audit}
\small
\begin{tabular}{rcc}
\toprule
$n$ & Vol. rel. err. & Max recon. err. \\
\midrule
2 & 0.00495 & $1.33{\times}10^{-15}$ \\
3 & 0.00372 & $9.99{\times}10^{-16}$ \\
5 & 0.00354 & $4.33{\times}10^{-15}$ \\
10& 0.02020 & $1.39{\times}10^{-13}$ \\
\bottomrule
\end{tabular}
\end{table}

For low-dimensional projected-density normalization, importance estimates are $1.0159$ (SE $0.0275$, effective sample size $812.6$) for $n=2$ and $0.9148$ (SE $0.0594$, ESS $212.2$) for $n=3$.  The same proposal degenerates for $n=5$ and $10$ (ESS $4.85$ and $1$), so those high-dimensional importance estimates are not interpreted as model failures.  Direct coordinate reconstruction remains accurate in those dimensions.  Fourier-probe comparisons between analytic density calculations and Monte Carlo simulation give mean absolute errors $0.0183$ for $n=2$ and $0.0421$ for $n=3$. Full volume/sampler, normal-fiber reconstruction, density-normalization,
and analytic-versus-Monte-Carlo diagnostics are given in
Appendices~\ref{app:uniform-sampler}-\ref{app:density-validation}. 

\subsection{Parameter recovery and intrinsic energy}
A restricted isotropic recovery experiment uses truth $\alpha=0.75$ and $\sigma=0.8$.  In $n=2$, RMSE decreases for both parameters over the investigated sample sizes.  In $n=3$, location recovery improves while finite-sample scale RMSE is not monotone; this experiment is a numerical validation of Theorem~\ref{thm:mle}. Figure~\ref{fig:recovery-runtime} summarizes the parameter-recovery
behavior and the computational scaling of the sampler and projected
density evaluator. The complete bias and RMSE table is reported in
Appendix~\ref{app:recovery-audit}.

\begin{figure*}[t]
\centering
\begin{minipage}[t]{0.48\textwidth}
\centering
\includegraphics[width=\linewidth]{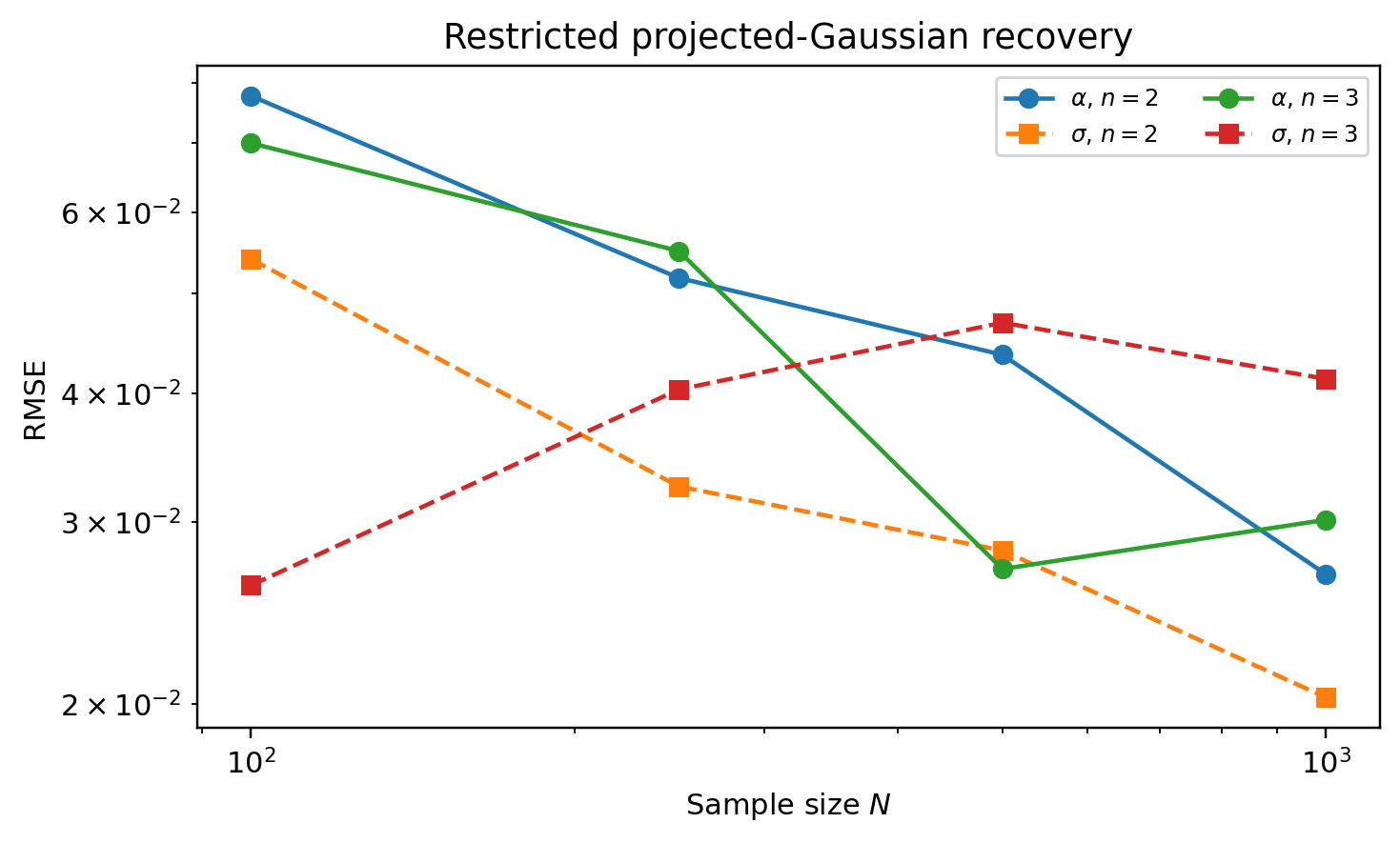}
\end{minipage}\hfill
\begin{minipage}[t]{0.48\textwidth}
\centering
\includegraphics[width=\linewidth]{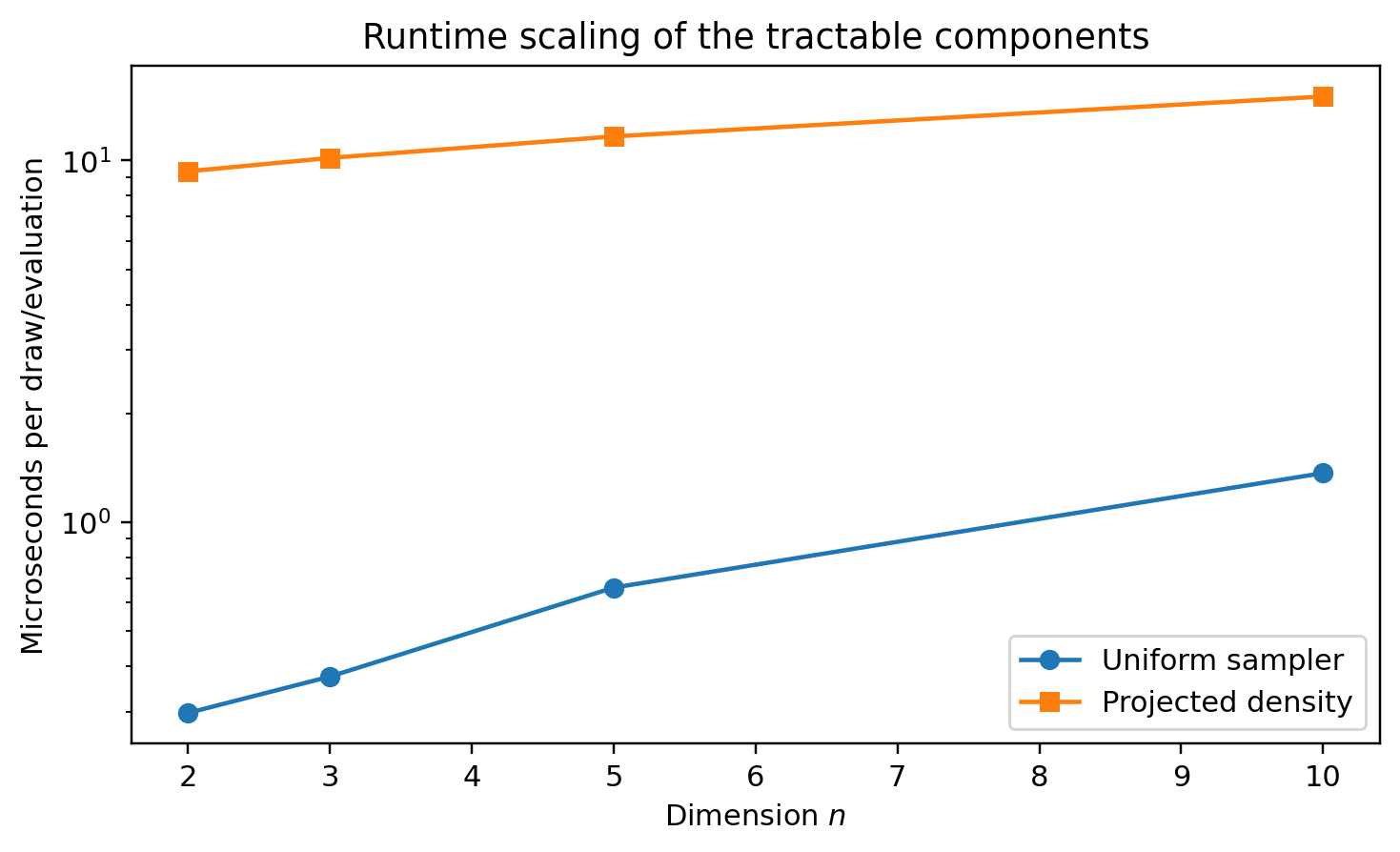}
\end{minipage}
\caption{\textbf{Left:} Restricted projected-Gaussian parameter-recovery study.  \textbf{Right:} Scaling of the exact surface-uniform sampler and projected-density evaluation.  Density evaluation remains a one-dimensional fiber calculation as $n$ grows.}
\label{fig:recovery-runtime}
\end{figure*}

The intrinsic energy solver is checked by an independent finite-difference gradient comparison (relative error $8.76\times10^{-11}$), mesh refinement from 16 to 32 path segments, multiple winding classes, and final gradient norms.  All six endpoint pairs in $n=2,3,5$ satisfy the prespecified $1\%$ mesh-refinement criterion and $\|\nabla E\|_\infty<10^{-4}$.  The worst mesh change is $0.7854\%$ and the worst final gradient norm is $1.76\times10^{-5}$.  Unlike the sampler and density evaluation, global path-energy optimization becomes expensive with dimension; Sampling and one-dimensional projected-density evaluation remain
computationally inexpensive over the investigated dimensions, whereas
intrinsic-energy optimization is substantially more expensive because
it combines path optimization with a winding search. 
Per-pair mesh-refinement results, gradient residuals, and selected
winding classes are reported in Appendix~\ref{app:energy-audit}.


\section{REAL-DATA APPLICATIONS}
\label{sec:realdata}
\subsection{Common comparison protocol}
\label{sec:realdata_protocol}

Each fitted model is evaluated by comparing samples generated from the
fitted model with the held-out real observations.  Both samples are
mapped through the common sine-cosine embedding
$$
\Phi(\theta)
=
(\cos\theta_1,\ldots,\cos\theta_d,
 \sin\theta_1,\ldots,\sin\theta_d)
\in\mathbb{R}^{2d}.
$$
We report Energy distance \cite[see][]{SzekelyRizzo2013}, multiscale
$\MMD^2$ \cite[see][]{Gretton2012}, and sliced Wasserstein distance (SWD)
\cite[see][]{Bonneel2015}; lower values indicate closer agreement between
the fitted model and the held-out real-data distribution. Since the competing
estimators use different fitting objectives, comparisons are based on
these common held-out sample discrepancies rather than a universal
likelihood criterion.

We compare six models: \textbf{M0}, independent marginal von Mises
maximum likelihood; \textbf{M1}, multivariate wrapped normal fitted by
latent wrapping-vector EM \cite{Greco2023}; \textbf{M2}, multivariate von Mises distribution fitted by maximum
pseudo-likelihood \cite{Mardia2008}; \textbf{M3}, the general
multivariate projected normal of \citet{Mastrantonio2018}, fitted by
latent-radius Bayesian MCMC with pairwise scale identification;
\textbf{M4}, our proposed projected
normal distributions on the curved torus fitted by minimum MMD; and \textbf{M5}, the interpretable
toroidal projected normal of \citet{KatoMastrantonioIshikawa2025},
fitted by latent-variable Bayesian MCMC.
For M3 and M5, held-out
evaluation uses posterior-predictive samples from retained MCMC draws.
All tuning and restart selection use validation data only, and grouped
or trajectory data are split at the group or trajectory level to avoid
frame-level leakage.

\subsection{Franka Panda lift-hold joint configurations}
\label{sec:franka}

We first analyze the public Franka Panda \texttt{lift\_hold}
trajectories from the Neural Actuation Dataset \cite{Dou2026}. The
source contains 30 official training trajectories and 5 official test
trajectories across five payloads. Trial \texttt{007} from each payload
is reserved from the training set for validation, giving 25 fitting,
5 validation, and 5 untouched official-test trajectories. The coordinates are kept in their physical serial-chain order
$(q_1,q_2,q_3,q_4)$; no permutation of the joints is selected using
validation or test performance.

The M4 geometry is fixed before official-test evaluation at
$$
R=(1,1.05,2.1525,4.412625),
$$
with covariance floor $10^{-3}$. Three M4 restarts are selected only
by validation MMD. The competing model definitions and estimation
procedures are unchanged across applications.

As shown in Table~\ref{tab:franka_m0_m5}, M4 attains the lowest
mean Energy and $\MMD^2$ discrepancies among all six models.. Relative to the strongest non-M4 competitor for each of
these criteria, the reduction is approximately $22.6\%$ in Energy
(M5) and $26.6\%$ in $\MMD^2$ (M3). The trajectory-wise comparison is
also consistent: M4 has lower Energy and $\MMD^2$ than every competitor
on all five official test trajectories. In contrast, M1 attains the
smallest SWD, and M5 also improves on M4 under SWD. Thus the Franka
analysis supports a strong but metric-dependent benefit from the
recursive geometry rather than uniform dominance. Additional paired trajectory results and model-specific numerical
diagnostics are reported in
Appendices~\ref{app:paired-robot} and~\ref{app:realdata-diagnostics}.

\begin{table}[h!]
\vspace{-1em}
\centering
\caption{Franka Panda official-test discrepancies, reported as
mean (SE) across the five held-out trajectories. Lower is better.}
\label{tab:franka_m0_m5}
\small
\setlength{\tabcolsep}{2.6pt}
\scalebox{1}{\begin{tabular}{lcccc}
\toprule
Model & Params. & Energy & $\MMD^2$ & SWD\\
\midrule
M0 & 8  & 0.05828 (0.00180) & 0.15284 (0.00495) & 0.04653 (0.00058)\\
M1 & 14 & 0.05500 (0.00162) & 0.14151 (0.00294) & \textbf{0.04344 (0.00029)}\\
M2 & 14 & 0.06350 (0.00242) & 0.15397 (0.00357) & 0.04420 (0.00060)\\
M3 & 40 & 0.08175 (0.00304) & 0.08289 (0.00323) & 0.09314 (0.00263)\\
M4 & 15 & \textbf{0.03752 (0.00096)} & \textbf{0.06080 (0.00084)} & 0.04766 (0.00102)\\
M5 & 14 & 0.04849 (0.00081) & 0.12124 (0.00217) & 0.04405 (0.00041)\\
\bottomrule
\end{tabular}}
\vspace{-1em}
\end{table}

\subsection{ALOHA coffee manipulation}
\label{sec:aloha_coffee}

As a second ordered-kinematic application, we analyze the public real
ALOHA static-coffee demonstrations associated with the ALOHA platform
\cite{Zhao2023}. The dataset contains 50 episodes recorded at 50 Hz.
We use the first four revolute joints of the left serial arm,
$
(\texttt{left\_waist},
 \texttt{left\_shoulder},
 \texttt{left\_elbow},
$ and $ \texttt{left\_forearm\_roll})
\in\T^4,$
exclude the gripper coordinate and retain every fifth frame. These coordinates are retained in their physical order along the left
serial arm, so the recursive ordering is specified by robot kinematics
rather than by empirical model selection.
A fixed
episode-level split gives 35 fitting, 7 validation, and 8 untouched
test episodes. The same six model classes, M4 geometry, covariance
floor, optimization settings, and held-out discrepancy definitions
used for Franka are retained; only the data source is changed. See 
Table~\ref{tab:aloha_coffee} which shows a similar but more competitive
pattern compared to the Franka-Panda application.
\begin{table}[h!]
\vspace{-1em}
\centering
\caption{ALOHA coffee held-out discrepancies, reported as mean (SE)
across the eight test episodes. Lower is better.}
\label{tab:aloha_coffee}
\small
\setlength{\tabcolsep}{2.6pt}
\scalebox{1}{\begin{tabular}{lcccc}
\toprule
Model & Params. & Energy & $\MMD^2$ & SWD\\
\midrule
M0 & 8  & 0.12098 (0.01497) & 0.07257 (0.00984) & 0.09436 (0.00537)\\
M1 & 14 & 0.10439 (0.01479) & 0.05821 (0.00964) & 0.08339 (0.00584)\\
M2 & 14 & 0.10388 (0.01374) & 0.05725 (0.00886) & 0.08181 (0.00544)\\
M3 & 40 & 0.09897 (0.01531) & 0.05389 (0.00975) & \textbf{0.07991 (0.00608)}\\
M4 & 15 & \textbf{0.09718 (0.01355)} & \textbf{0.05062 (0.00775)} & 0.08259 (0.00567)\\
M5 & 14 & 0.10427 (0.01391) & 0.05784 (0.00932) & 0.08709 (0.00554)\\
\bottomrule
\end{tabular}}
\vspace{-1em}
\end{table}

 M4 gives the lowest mean Energy and $\MMD^2$, whereas the
more highly parameterized M3 gives the lowest SWD. Relative to M3, the M4
mean is lower by approximately $1.8\%$ in Energy and $6.1\%$ in
$\MMD^2$; the corresponding paired episode differences favor M4 on
6 of 8 episodes for each criterion. The SWD comparison reverses this
pattern: M3 is lower on 7 of 8 episodes. The comparison with the
similarly sized M5 model is more pronounced; M4 reduces the mean
Energy, $\MMD^2$, and SWD by approximately $6.8\%$, $12.5\%$, and
$5.2\%$, respectively. Because the M3-M4 differences in the ALOHA
analysis are modest relative to episode-to-episode variability, we
interpret them descriptively rather than as evidence of uniform
superiority. Paired episode-level comparisons are reported in
Appendix~\ref{app:paired-robot}.

\subsection{Cross-domain stress test: cysteine backbone dihedrals}
\label{sec:cysteine_stress}

To test whether the recursive geometry is advantageous outside the
ordered robot-kinematic setting, we retain the public cysteine
backbone-dihedral application of \citet{GarridoRuiz2022}. The
$(\phi,\psi)\in\T^2$ angles for CYS, CSO, CSD, and OCS are modeled
separately. For the bivariate backbone analysis we use the conventional
$(\phi,\psi)$ coordinate order.  Because the recursive model is
order-dependent, the reversed order $(\psi,\phi)$ is examined as a
supplementary sensitivity check and is not selected using the held-out
test data.
Within each state we use a 70/15/15
fit/validation/test split; when a PDB or structure identifier is
available, splitting is performed at structure level. The main summary
is the equal-weight macro-average over the four oxidation states, so
that the largest state does not dominate the conclusion.

\begin{table}[h]
\vspace{-1em}
\centering
\caption{Cysteine macro-averaged held-out discrepancies across CYS,
CSO, CSD, and OCS. Lower is better. Full state-wise results are given in
Appendix~\ref{app:cysteine-statewise}.}
\label{tab:cysteine_macro_m0_m5}
\small
\begin{tabular}{lccc}
\toprule
Model & Energy & $\MMD^2$ & SWD\\
\midrule
M0 & 0.16364 & 0.04714 & 0.20876\\
M1 & 0.15704 & 0.04447 & 0.20855\\
M2 & 0.13974 & 0.03833 & 0.19407\\
M3 & \textbf{0.08547} & \textbf{0.02075} & \textbf{0.14594}\\
M4 & 0.13119 & 0.03500 & 0.17664\\
M5 & 0.13425 & 0.03725 & 0.19395\\
\bottomrule
\end{tabular}
\vspace{-1em}
\end{table}

Table~\ref{tab:cysteine_macro_m0_m5} shows that M3 is clearly
strongest in the macro-average under all three discrepancies.
The full state-wise results in
Appendix~\ref{app:cysteine-statewise} show the same pattern in
each oxidation state. Relative to M4, the
macro-averaged discrepancies are smaller by approximately $34.9\%$,
$40.7\%$, and $17.4\%$, respectively. M4 remains second-best in the
macro-average under all three criteria, but the cysteine analysis
demonstrates that the recursive geometry is not universally preferred.

\subsection{Dependence on coordinate ordering.}
The recursive embedding is not permutation invariant, so the angular
ordering forms part of the model specification.  In the robot
applications we use the physical proximal-to-distal joint order,
independently of model fit.  A validation-only sensitivity analysis
over all $4!=24$ permutations shows that this physical order ranks
first for ALOHA Coffee and fifth for Franka lift-hold.  In contrast,
for cysteine backbone angles the reversed $(\psi,\phi)$ ordering gives
lower validation MMD than the conventional $(\phi,\psi)$ ordering. The complete validation-only ordering  is reported in
Appendix~\ref{app:ordering-sensitivity}.
These results indicate that the recursive construction is especially
natural when an application supplies an external sequential ordering;
otherwise the coordinate order should be treated as a modeling choice
and selected without access to the test data.

Taken together, the applications support a selective inductive-bias interpretation. M4 performs particularly well for ordered robot-joint data, attaining the lowest held-out Energy and $\MMD^2$ in both Franka and ALOHA, whereas the more flexible M3 is preferred for cysteine dihedrals. This suggests that the recursive geometry is most useful when the coordinate ordering reflects meaningful application structure.


\section{DISCUSSION AND CONCLUSION}

Due to the recursivity of the embedding $F_n$, the nested revolution structure yields an
orthogonal metric, an explicit normal field, and affine normal-fiber
coordinates. These properties reduce the projected Gaussian density
to a one-dimensional normal-fiber integral with a
polynomial-trigonometric Jacobian, while also providing a natural
intrinsic energy functional for geometric summaries.

The empirical results support a context-dependent
advantage.  The proposed model (M4) attains the lowest held-out Energy and $\MMD^2$ in both
robot applications, whereas the more flexible M3 is preferred for
cysteine dihedrals.  Together with the ordering-sensitivity results in
Appendix~\ref{app:ordering-sensitivity}, this suggests that the recursive
geometry is most useful when its ordered dependence is compatible with
application structure.

The geometry has several limitations.  The strong radius condition
\eqref{eq:radius} restricts the admissible embeddings, and global
parameter identifiability is established only for the isotropic
subfamily.  For the flexible covariance model, the theoretical guarantee
is distributional minimum-MMD consistency for the idealized
simulation-based criterion, with the finite-batch optimizer serving as
its practical approximation.  The construction is also not permutation
invariant, so coordinate ordering forms part of the model specification.
Finally, intrinsic-distance computation requires nonconvex optimization
over winding classes and is substantially more expensive than sampling
or density evaluation.  Extensions to weaker embedding conditions,
richer identifiable covariance structures, scalable ordering strategies,
and faster intrinsic solvers remain natural directions.

In summary, the proposed normal-fiber projected Gaussian framework
combines explicit curved-torus geometry, a tractable projected density,
likelihood theory in an identifiable restricted family, minimum-MMD
estimation, and intrinsic geometric statistics.  The numerical and
real-data results support the recursive construction as a situationally useful model.

\section*{Acknowledgement}

This project has been funded by the Deutsche Forschungsgemeinschaft (DFG, German Research Foundation) under Germany's Excellence Strategy – EXC-2047/2 – 390685813.

\bibliographystyle{apalike}
\bibliography{references}

@book{MardiaJupp2000,
  title={Directional Statistics},
  author={Mardia, Kanti V. and Jupp, Peter E.},
  year={2000},
  publisher={Wiley}
}

@article{Mastrantonio2018,
  author  = {Mastrantonio, Gianluca},
  title   = {The joint projected normal and skew-normal: {A} distribution for poly-cylindrical data},
  journal = {Journal of Multivariate Analysis},
  volume  = {165},
  pages   = {14--26},
  year    = {2018}
}

@article{KatoMastrantonioIshikawa2025,
  author  = {Kato, Shogo and Mastrantonio, Gianluca and Ishikawa, Masayuki},
  title   = {An interpretable family of projected normal distributions and a related copula model for {B}ayesian analysis of hypertoroidal data},
  journal = {arXiv preprint arXiv:2508.16432},
  year    = {2025}
}

@incollection{Diaconis2013,
  title={Sampling from a Manifold},
  author={Diaconis, Persi and Holmes, Susan and Shahshahani, Mehrdad},
  booktitle={Advances in Modern Statistical Theory and Applications: {A} Festschrift in Honor of {M}orris {L}. {E}aton},
  pages={102--125},
  year={2013},
  publisher={Institute of Mathematical Statistics}
}

@article{BanerjeeBiswas2026,
  title={Intrinsic geometry-inspired dependent toroidal distribution: {A}pplication to regression model for astigmatism data},
  author={Banerjee, Buddhananda and Biswas, Surojit},
  journal={Computational Statistics},
  volume={41},
  pages={37},
  year={2026}
}

@article{Greco2023,
  title={Finite mixtures of multivariate wrapped normal distributions for model based clustering of p-torus data},
  author={Greco, Luca and Novi Inverardi, Pierpaolo and Agostinelli, Claudio},
  journal={Journal of Computational and Graphical Statistics},
  volume={32},
  number={3},
  pages={1215--1228},
  year={2023}
}

@article{Mardia2008,
  title={A multivariate von {M}ises distribution with applications to bioinformatics},
  author={Mardia, Kanti V. and Hughes, Gareth and Taylor, Charles C. and Singh, Harpreet},
  journal={Canadian Journal of Statistics},
  volume={36},
  number={1},
  pages={99--109},
  year={2008}
}

@article{WangGelfand2013,
  title={Directional data analysis under the general projected normal distribution},
  author={Wang, Fangpo and Gelfand, Alan E.},
  journal={Statistical Methodology},
  volume={10},
  number={1},
  pages={113--127},
  year={2013}
}

@article{Gretton2012,
  title={A kernel two-sample test},
  author={Gretton, Arthur and Borgwardt, Karsten M. and Rasch, Malte J. and Sch{\"o}lkopf, Bernhard and Smola, Alexander},
  journal={Journal of Machine Learning Research},
  volume={13},
  pages={723--773},
  year={2012}
}

@inproceedings{Zhao2023,
  author    = {Zhao, Tony Z. and Kumar, Vikash and Levine, Sergey and Finn, Chelsea},
  title     = {Learning fine-grained bimanual manipulation with low-cost hardware},
  booktitle = {Robotics: Science and Systems},
  year      = {2023}
}

@book{srivastava2016functional,
  author    = {Srivastava, Anuj and Klassen, Eric P.},
  title     = {Functional and Shape Data Analysis},
  series    = {Springer Series in Statistics},
  publisher = {Springer},
  address   = {New York},
  year      = {2016}
}

@article{Singh2002Probabilistic,
  author    = {Singh, Harshinder and Hnizdo, Vladimir and Demchuk, Eugene},
  title     = {Probabilistic model for two dependent circular variables},
  journal   = {Biometrika},
  volume    = {89},
  number    = {3},
  pages     = {719--723},
  year      = {2002},
  publisher = {Oxford University Press}
}

@article{Afsari2011,
  title={Riemannian $L^p$ center of mass: {E}xistence, uniqueness, and convexity},
  author={Afsari, Bijan},
  journal={Proceedings of the American Mathematical Society},
  volume={139},
  number={2},
  pages={655--673},
  year={2011}
}

@article{GarridoRuiz2022,
  title={Cysteine oxidation in proteins: {S}tructure, biophysics, and simulation},
  author={Garrido Ruiz, Diego and Sandoval-Perez, Ana and Rangarajan, A. V. and Gunderson, E. L. and Jacobson, Matthew P.},
  journal={Biochemistry},
  volume={61},
  number={20},
  pages={2165--2176},
  year={2022}
}

@inproceedings{Dou2026,
  title={NeuralActuator: Neural Actuation Modeling for Robot Dynamics and External Force Perception},
  author={Dou, Zhiyang and Onyemelukwe, John U. and Zhang, Hangxing and Zhang, Heng and Guo, Minghao and Tian, Yunsheng and Lipiec, Michal Piotr and Jacob, Joshua and Liu, Chao and Chen, Peter Yichen and Ivanov, Yuri and Matusik, Wojciech},
  booktitle={Robotics: Science and Systems},
  year={2026}
}

@article{SzekelyRizzo2013,
  title={Energy Statistics: {A} Class of Statistics Based on Distances},
  author={Sz{\'e}kely, G{\'a}bor J. and Rizzo, Maria L.},
  journal={Journal of Statistical Planning and Inference},
  volume={143},
  number={8},
  pages={1249--1272},
  year={2013}
}

@article{Bonneel2015,
  title={Sliced and {R}adon {W}asserstein Barycenters of Measures},
  author={Bonneel, Nicolas and Rabin, Julien and Peyr{\'e}, Gabriel and Pfister, Hanspeter},
  journal={Journal of Mathematical Imaging and Vision},
  volume={51},
  number={1},
  pages={22--45},
  year={2015}
}

@article{bhattacharya2003large,
  title={Large sample theory of intrinsic and extrinsic sample means on manifolds},
  author={Bhattacharya, Rabi and Patrangenaru, Vic},
  journal={The Annals of Statistics},
  volume={31},
  number={1},
  pages={1--29},
  year={2003},
  publisher={Institute of Mathematical Statistics}
}

@article{llibre2024darboux,
  title={Darboux theory of integrability in $\mathbb{T}^n$},
  author={Llibre, Jaume and Valls, Claudia},
  journal={European Journal of Mathematics},
  volume={10},
  number={2},
  pages ={26},
  year={2024},
  publisher={Springer}
}

@article{nodehi2021estimation,
  title={Estimation of parameters in multivariate wrapped models for data on ap-torus},
  author={Nodehi, Anahita and Golalizadeh, Mousa and Maadooliat, Mehdi and Agostinelli, Claudio},
  journal={Computational Statistics},
  volume={36},
  number={1},
  pages={193--215},
  year={2021},
  publisher={Springer}
}

@article{mardia2014some,
  title={Some fundamental properties of a multivariate von Mises distribution},
  author={Mardia, Kanti V and Voss, Jochen},
  journal={Communications in Statistics-Theory and Methods},
  volume={43},
  number={6},
  pages={1132--1144},
  year={2014},
  publisher={Taylor \& Francis}
}

@article{xu2023density,
  title={Density estimation for toroidal data using semiparametric mixtures},
  author={Xu, Danli and Wang, Yong},
  journal={Statistics and Computing},
  volume={33},
  number={6},
  pages={140},
  year={2023},
  publisher={Springer}
}

@article{nodehi2025torus,
  title={Torus probabilistic principal component analysis},
  author={Nodehi, Anahita and Golalizadeh, Mousa and Maadooliat, Mehdi and Agostinelli, Claudio},
  journal={Journal of Classification},
  volume={42},
  number={2},
  pages={435--456},
  year={2025},
  publisher={Springer}
}

@article{Eltzner2018torus,
    author= {Benjamin Eltzner and Stephan Huckemann and Kanti V. Mardia},
    title={Torus principal component analysis with applications to RNA structure}, 
    journal={The Annals of Applied Statistics},
    volume={12},
    number={2},
    pages={1332-1359},
    year={2018}
}

@article{hong2022clustorus,
  title={ClusTorus: An R Package for Prediction and Clustering on the Torus by Conformal Prediction.},
  author={Hong, Seungki and Jung, Sungkyu},
  journal={R J.},
  volume={14},
  number={2},
  pages={187--208},
  year={2022}
}

@article{klein2020torus,
  title={Torus graphs for multivariate phase coupling analysis},
  author={Klein, Natalie and Orellana, Josue and Brincat, Scott L and Miller, Earl K and Kass, Robert E},
  journal={The annals of applied statistics},
  volume={14},
  number={2},
  pages={635},
  year={2020}
}

@article{gottard2021graphical,
  title={Graphical models for circular variables},
  author={Gottard, Anna and Panzera, Agnese},
  journal={arXiv preprint arXiv:2104.03194},
  year={2021}
}

@incollection{mescher2024geometric,
  title={Geometric and topological properties of manifolds in robot motion planning},
  author={Mescher, Stephan},
  booktitle={Topology and AI},
  pages={123--159},
  year={2024},
  publisher={European Mathematical Society-EMS-Publishing House GmbH}
}

@inproceedings{sato2023online,
  title={Online Estimation of Self-Body Deflection With Various Sensor Data Based on Directional Statistics},
  author={Sato, Hiroya and Kawaharazuka, Kento and Makabe, Tasuku and Okada, Kei and Inaba, Masayuki},
  booktitle={2023 IEEE/SICE International Symposium on System Integration (SII)},
  pages={1--8},
  year={2023},
  organization={IEEE}
}

@article{giannakis2019data,
  title={Data-driven spectral decomposition and forecasting of ergodic dynamical systems},
  author={Giannakis, Dimitrios},
  journal={Applied and Computational Harmonic Analysis},
  volume={47},
  number={2},
  pages={338--396},
  year={2019},
  publisher={Elsevier}
}

@article{biswas2026efficient,
  title={Efficient sampling from circular distributions: extensions to toroidal and spherical distributions},
  author={Biswas, Surojit and Banerjee, Buddhananda},
  journal={Communications in Statistics-Simulation and Computation},
  pages={1--24},
  year={2026},
  publisher={Taylor \& Francis}
}

\appendix


\section{Theory and proofs}
Throughout, we shall assume that $n\ge 2$ and that the radii $R_1,\ldots,R_n>0$ satisfy the strong-radius condition \eqref{eq:radius}.

\subsection{Proof of Theorem~\ref{thm:exactuniform_alg}}
\label{sec:proof_exactuniform_alg}

The main paper uses the recursively embedded torus similarly to \cite{llibre2024darboux}
\[
\rho_1=R_1,\quad \rho_k=R_k+\rho_{k-1}\cos\theta_{k-1},
\]
\[
F_n({\bm{\theta}})=
(\rho_n\cos\theta_n,\rho_n\sin\theta_n,
\rho_{n-1}\sin\theta_{n-1},\ldots,\\
\rho_1\sin\theta_1),    
\]

under $R_k>\sum_{j<k}R_j$.  Direct differentiation gives
\begin{align*}
g&=\operatorname{diag}(\rho_1^2,\ldots,\rho_n^2),&
\quad\operatorname{dVol}_n&=\prod_{k=1}^n\rho_k\,\dd\theta_k,\\
\operatorname{Vol}(F_n(\T^n))&=(2\pi)^n\prod_{k=1}^nR_k.    
\end{align*}

For the ordinary $2$-torus, the corresponding area-uniform angular law is classical; see \citet{Diaconis2013}.  The following lemma gives the recursive formula for the volume form.

\begin{lemma}
\label{lem:uniform_prob_density}
If $Y$ is uniform on $F_n(\T^n)$ with respect to normalized Riemannian volume, its angular coordinate $\Theta$ has density
\[
p_U({ \bm{\theta}})
=
\frac1{(2\pi)^n}
\prod_{j=1}^{n-1}
\left(1+\frac{\rho_j}{R_{j+1}}\cos\theta_j\right).
\]
Equivalently,
\[
\Theta_j\mid\Theta_{1:j-1}
\sim
\frac1{2\pi}\left(1+a_j\cos\theta_j\right),
\quad
a_j=\frac{\rho_j}{R_{j+1}},
\]
for $j=1,\ldots,n-1$, and $\Theta_n\sim U(0,2\pi)$.
\end{lemma}

\begin{proof}
Normalize the Riemannian volume element:
\[
p_U(\theta)=
\frac{\prod_{k=1}^n\rho_k}
{(2\pi)^n\prod_{k=1}^nR_k}.
\]
Since $\rho_1=R_1$ and
$\rho_{j+1}/R_{j+1}=1+(\rho_j/R_{j+1})\cos\theta_j$,
the product factorization follows.  Each $\rho_j$ depends only on
$\theta_1,\ldots,\theta_{j-1}$, hence the factors are valid conditional densities.
Condition $R_{j+1}>\sum_{k\le j}R_k$ implies $0<a_j<1$.
\end{proof}

\begin{proof}[Proof of Theorem~\ref{thm:exactuniform_alg}]
The transformation at stage $j$ generates the required Cardioid conditional
$(2\pi)^{-1}(1+a_j\cos\theta_j)$.  Multiplying the sequential conditionals and the final uniform density gives exactly the density in Lemma \ref{lem:uniform_prob_density}. Thus the exact cardioid transformation due to \citet{BanerjeeBiswas2026} can be applied sequentially, giving an exact surface-uniform sampler of computational complexity  {$\mathcal{O}(n)$}. 
\end{proof}

\subsection{A normal-fiber coordinate lemma}
\label{sec:normalfibercoordinatelemma}

In this section we introduce a lemma that establishes fundamental properties for the normal-fiber coordinates given by \eqref{eq:fiber-map}.
Set
\[
\mathcal D
=
\left\{(T,{\bm{\theta}}):{\bm{\theta}}\in\T^n,\ 0<T<T_0({\bm{\theta}})\right\}.
\]
By construction, $(T,{\bm{\theta}})\in\mathcal D$ if and only if
\[
T>0,
\quad \text{ and } \quad
q_k(T,{\bm{\theta}})>0,
\quad k=2,\ldots,n.
\]

\begin{lemma}
\label{lem:fiber-coordinates}
The following hold:

\begin{enumerate}
\item The map $X:\mathcal D\to\R^{n+1}$ is smooth and its pullback Euclidean metric is given by \eqref{eq:fiber-metric} and the Jacobian of $X$ is given by \eqref{eq:fiber-jac}.

\item The vector $\partial_T X(T,{\bm{\theta}})$ has unit norm and is orthogonal to every $\partial_{\theta_k}X(T,{\bm{\theta}})$. In particular,
\[
\nu({\bm{\theta}}):=\partial_T X(T,{\bm{\theta}})
\]
is independent of $T$ and is the unit normal vector field given by \eqref{eq:normal}.

\item {We} have the affine normal-fiber representation \eqref{eq:fiber-affine} holds true.

\item The coordinate map is one-to-one modulo the usual $2\pi$ identification of each angular coordinate on its image, and it admits an explicit inverse away from a Lebesgue-null singular set.
\end{enumerate}
\end{lemma}

\begin{proof}
The metric identity is most transparent by induction over the dimension. In dimension 1, the ordinary polar transformation has metric
\[
\dd T^2+T^2\dd\theta_1^2.
\]
For dimension $k\ge2$, the pair
\[
q_k=R_k+q_{k-1}\cos\theta_{k-1},
\quad
y_k=q_{k-1}\sin\theta_{k-1}
\]
satisfies
\[
\dd q_k^2+\dd y_k^2
=
\dd q_{k-1}^2+q_{k-1}^2\dd\theta_{k-1}^2.
\]
Adding the final rotation by $\theta_k$ contributes $q_k^2\dd\theta_k^2$. Repeating this argument yields \eqref{eq:fiber-metric}. Hence the metric tensor is diagonal with diagonal entries
\[
1,
\ T^2,
\ q_2^2,
\ldots,
\ q_n^2,
\]
which gives \eqref{eq:fiber-jac}.

The same diagonal metric shows that $\partial_T X$ has unit norm and is orthogonal to all angular coordinate vectors. Since $q_k$ is affine in $T$, differentiating \eqref{eq:fiber-map} with respect to $T$ gives the explicit vector
\[
\nu({{\bm{\theta}}})=
\big(
\textstyle\prod_{j=1}^{n}\cos\theta_j,
\ (\prod_{j=1}^{n-1}\cos\theta_j)\sin\theta_n,
\ldots,
\cos\theta_1\sin\theta_2,
\sin\theta_1
\big)^\top,    
\]
which is independent of $T$. Integrating $\partial_T X=\nu(\theta)$ in $T$ and evaluating at $T=R_1$ proves \eqref{eq:fiber-affine}.

For the inverse, let  {$\bm{x}=(x_0,\ldots,x_{n})^\top\in\R^{n+1}$}. Define recursively
\[
\widehat q_n={(x_0^2+x_1^2)^{1/2} },
\quad
\widehat\theta_n=\operatorname{atan2}{(x_1,x_0)},
\]
and, for $k=n-1,\ldots,1$,
\[
\widehat q_k
=
\left\{(\widehat q_{k+1}-R_{k+1})^2+{x_{n-k+1}^2}\right\}^{1/2},
\label{eq:inverse-q}
\]
\[
\widehat\theta_k
=
\operatorname{atan2}
\left({x_{n-k+1}},\widehat q_{k+1}-R_{k+1}\right).
\]
Whenever all $\widehat q_k>0$, these relations uniquely recover
\[
T=\widehat q_1,
\quad
{\bm{\theta}}=(\widehat\theta_1,\ldots,\widehat\theta_n)
\]
modulo the standard angular identification. Substituting the recovered quantities into \eqref{eq:fiber-map} reconstructs {$\bm{x}$} exactly.

The singular set is
\[
\mathcal C
=
\bigcup_{k=1}^n\{x:\widehat q_k(x)=0\}.
\]
For $k=n$, $\widehat q_n=0$ implies $x_1=x_2=0$. For $k<n$, $\widehat q_k=0$ implies in particular {$x_{n-k+1}=0$}. Thus each component of $\mathcal C$ lies in a coordinate hyperplane and has Lebesgue measure zero. Hence $\mathcal C$ is Lebesgue-null. This proves the a.e. validity and injectivity of the normal-fiber coordinates.
\end{proof}

\subsection{Proof of Theorem~\ref{thm:pgdensity}}
\label{sec:proof_prob_density}

Let
\[
Z\sim\N_{n+1}(\mu,\Sigma),
\quad \Sigma\succ0.
\]
Since a non-degenerate Gaussian is absolutely continuous, Lemma~\ref{lem:fiber-coordinates} implies that $Z\notin\mathcal C$ almost surely.

\begin{proof}[Proof of Theorem~\ref{thm:pgdensity}]
{
Classically, it holds that $Z$ has probability density function in Euclidean standard coordinates at $\bm{x}= (x_0,\dots,x_n)\in \mathbb{R}^n$
\[
p_Z(\bm{x}) = \frac{1}{(2\pi)^{(n+1)/2}|\Sigma|^{1/2}} \\\exp\left(-\frac{1}{2}(\bm{x}-\mu)^T\Sigma^{-1}(\bm{x}-\mu)\right).    
\]
By the change of variables $X(T,\bm{\theta})= \bm{x}$ and by formula \eqref{eq:fiber-jac} it holds that $P_Z(T,Z)$ can be written in tubular normal coordinates of the form
\[
    p_Z(T,\bm{\theta}) =\frac{1}{(2\pi)^{(n+1)/2}|\Sigma|^{1/2}} \exp\Big(-\frac{1}{2}(X(T,\bm{\theta})-\mu)^T\Sigma^{-1}(X(T,\bm{\theta})-\mu)\Big) T\prod_{k=2}^n q_k(T,\theta).
\]
By rewriting the coordinate map in the affine normal representation \eqref{eq:fiber-affine} we can split the exponential into a $T$-dependent and a $T$-independent part. Hence, by defining $a(\bm{\theta})= F_n(\bm{\theta})-R_1\nu(\bm{\theta})-\mu$ we have that
\begin{multline*}
\exp\left(-\frac{1}{2}(X(T,\bm{\theta})-\mu)^\top\Sigma^{-1}(X(T,\bm{\theta})-\mu)\right)\\
 = \exp\Big(-\frac{1}{2}(F_n(\bm{\theta})-R_1\nu(\bm{\theta})+T\nu(\bm{\theta})-\mu)^\top \Sigma^{-1}(F_n(\bm{\theta})-R_1\nu(\bm{\theta})+T\nu(\bm{\theta})-\mu)\Big)\\
 = \exp\left(-\frac{1}{2}(a(\bm{\theta})+T\nu(\bm{\theta}))^\top\Sigma^{-1}(a(\bm{\theta})+T\nu(\bm{\theta}))\right)\\
 =\exp(C(\bm{\theta})) \exp\left(-\frac12A(\bm{\theta})T^2-B(\bm{\theta})T\right)
\end{multline*}
where we have used that $\Sigma^{-1}$ is symmetric. Note that by definition of the projected normal the marginal distribution can be obtained from the probability density function of 
\[
p_{\pr(Z)}(\bm{\theta}) = \int_0^{T_0(\bm(\theta))} p_Z(T,\bm{\theta}) \dd T
\]
Therefore, 
\[
    p_{\operatorname{pr}(Z)}(\bm{\theta})=\frac{\exp(C(\bm{\theta}))}
{(2\pi)^{(n+1)/2}|\Sigma|^{1/2}}
\int_0^{T_0(\bm{\theta})}
\exp\left(-\frac12A(\bm{\theta})T^2-B(\bm{\theta})T\right)
T\prod_{k=2}^n q_k(T,\bm{\theta})\,\dd T.
\]
Now through a change of variables $S=\sqrt{A}(\bm{\theta})T+K(\bm{\theta})$ we obtain \eqref{eq:angular-density}. Since this is the Radon-Nikodym derivative w.r.t the measure $\dd \bm{\theta}$ and because $\operatorname{dVol}(\bm{\theta}) = J(R_1,\bm{\theta}) \dd \bm{\theta}$ by \eqref{eq:fiber-jac}, we obtain \eqref{eq:surface-density}.

}

\end{proof}

\subsection{The projected normal in truncated moment form}
\label{sec:truncated_form_pg}

{Set $\beta_1=0$ then we may inductively define for $k\ge2$,
\[
\beta_k(\bm{\theta})=R_k+\beta_{k-1}(\bm{\theta})\cos(\theta_{k-1}),
\]
Then we can write out the explicit $T$-dependence of $q_k$ as
\begin{equation}
q_k(T,\bm{\theta})=\beta_k(\bm{\theta})+T\prod_{j=1}^{k-1}\cos\theta_j.
\label{eq:q-affine}
\end{equation}
Moreover,
\[
\beta_k(\bm{\theta})
\ge R_k-\sum_{j=2}^{k-1}R_j>0,
\]
as a consequence of \eqref{eq:radius}. Hence, if $T_0(\bm{\theta})$ is finite then 
\[
T_0(\bm{\theta}) = \min_{k\geq 2} -\frac{\beta_k(\bm{\theta})}{\prod_{j=1}^{k-1}\cos\theta_j}.
\]
Therefore, the integrand in \eqref{eq:angular-density} is an can be expressed as the standard normal Guassian times a polynomial of order $n$, that is the integrand can be written as
\begin{equation}
\exp\left(-\frac{1}{2}S^2\right)
(S-K(\bm{\theta}))\prod_{k=2}^n q_k\left(\frac{S-K(\bm{\theta})}{\sqrt{A}(\bm{\theta})},\bm{\theta}\right) 
=\exp\left(-\frac{1}{2}S^2\right)\sum_{m=0}^{n}c_m(\bm{\theta})S^m.
\label{eq:polynomial_notation_gaussian_integrand}
\end{equation}
Now each integral $\int \frac{1}{\sqrt{2\pi}}\exp(-S^2/2) S^m \dd S$ has a known closed form with
\[
I(0,a,b) := \int_a^b  \frac{1}{\sqrt{2\pi}}\exp(-S^2/2) = \Phi(b)-\Phi(a)
\] 
where $\Phi$ denotes the cumulative distribution function for the standard normal, and for integers $m\geq 1$
\begin{multline}
I(m,a,b):= \int_a^b \frac{\exp(-S^2/2)}{\sqrt{2\pi}} S^m \dd S\\ = \begin{cases}
	    \frac{\exp(-a^2/2)}{\sqrt{2\pi}} \sum_{j=0}^{(m-1)/2} \frac{(m-1)!!}{(2j)!!} a^{2j}  - \frac{\exp(-b^2/2)}{\sqrt{2\pi}} \sum_{j=0}^{(m-1)/2} \frac{(m-1)!!}{(2j)!!} b^{2j}, &  m \text{ odd},\\
	    \frac{\exp(-a^2/2)}{\sqrt{2\pi}} \sum_{j=0}^{(m-2)/2} \frac{(m-1)!!}{(2j+1)!!} a^{2j+1} - \frac{\exp(-b^2/2)}{\sqrt{2\pi}} \sum_{j=0}^{(m-2)/2} \frac{(m-1)!!}{(2j+1)!!} b^{2j+ 1} +(m-1)!!(\Phi(b)-\Phi(a)), 
 &  m \text{ even},
\end{cases}
\label{eq:guassian_moments}
\end{multline}
where $k!!$ denotes the double factorial which factors all integers up to $k$ with the same parity as $k$. Therefore, from \eqref{eq:guassian_moments}, \eqref{eq:angular-density} and \eqref{eq:polynomial_notation_gaussian_integrand} we can conclude that
\begin{equation}
p_{\operatorname{pr}(Z)}(\bm{\theta})
=
\frac{\exp\big(C(\bm{\theta}) + \frac{1}{2} K^2(\bm{\theta})\big)}
{(2\pi)^{n/2}|\Sigma|^{1/2} A(\bm{\theta})} \sum_{m=0}^n c_m I(m,K(\bm{\theta}),\sqrt{A}(\bm{\theta})T_0(\bm{\theta})+K(\bm{\theta})).
\label{eq:density-moment-reduction}
\end{equation}
}

\begin{remark}
For numerical work, the recurrence may be evaluated in a scaled or log-stabilized form when {$K^2(\bm{\theta})$} is large.
\end{remark}

\subsection{Proof of Theorem~\ref{thm:ident}}
\label{sec:proof_indentifiability}

Note that the probability density function \eqref{eq:angular-density} is valid for arbitrary positive-definite $\Sigma$. However, in  the case of Theorem~\ref{thm:ident}, we consider the isotropic family
\[
Z\sim\N_{n+1}(\mu,\sigma^2I_{n+1}),
\quad \mu\in\R^{n+1},\quad \sigma>0.
\]

\begin{proof}[Proof of Theorem~\ref{thm:ident}]
Write
\[
\lambda_j=\frac{\mu_j}{\sigma},
\quad j=0,\ldots,n.
\]
Because $q_n>0$ almost surely,
\[
\operatorname{sign}(Z_0)=\operatorname{sign}(\cos\Theta_n),
\quad
\operatorname{sign}(Z_1)=\operatorname{sign}(\sin\Theta_n).
\]
Therefore the angular law determines
\[
\mathbb P(\cos\Theta_n>0)
=\Phi(\lambda_0),
\quad
\mathbb P(\sin\Theta_n>0)
=\Phi(\lambda_1).
\]
Since $\Phi$ is strictly increasing, $\lambda_0$ and $\lambda_1$ are uniquely identified.

Now define
\[
Q_n=(Z_0^2+Z_1^2)^{1/2}.
\]
From the inverse-coordinate relation,
\[
Q_n-R_n=q_{n-1}\cos\Theta_{n-1}.
\]
Since $q_{n-1}>0$ almost surely,
\begin{equation}
\{\cos\Theta_{n-1}<0\}
=
\{Q_n<R_n\}
\quad\text{a.s.}
\label{eq:scale-event}
\end{equation}
Moreover,
\[
\frac{Q_n}{\sigma}
\sim
\left\|Y\right\|,
\quad
Y\sim\N_2\left(
\begin{bmatrix}\lambda_0\\ \lambda_1\end{bmatrix},I_2
\right).
\]
Let
\[
H_{\lambda_0,\lambda_1}(r)
=
\mathbb P\left(
\left\|Y\right\|<r
\right),
\quad r>0.
\]
This is the CDF of a noncentral chi random variable with two degrees of freedom; its density is strictly positive for every $r>0$. Hence $H_{\lambda_0,\lambda_1}$ is strictly increasing on $(0,\infty)$. The angular law therefore determines
\[
p
:=
\mathbb P(\cos\Theta_{n-1}<0)
=
H_{\lambda_0,\lambda_1}\left(\frac{R_n}{\sigma}\right).
\]
Since $0<p<1$ and $R_n>0$ is known,

{\begin{equation}
\sigma
=
\frac{R_n}
{H_{\lambda_0,\lambda_1}^{-1}(p)}.
\label{eq:sigma-identification}
\end{equation}}
Thus $\sigma$ is uniquely identified.

It remains to identify $\mu_2,\ldots,\mu_{n}$. For $j=2,\ldots,n$, let
\[
k=n-j+1.
\]
The forward coordinates satisfy
\[
Z_j=q_k\sin\Theta_k,
\]
with $q_k>0$ almost surely. Hence
\[
\mathbb P(\sin\Theta_k>0)
=
\mathbb P(Z_j>0)
=
\Phi\left(\frac{\mu_j}{\sigma}\right)
=
\Phi(\lambda_j).
\]
Again, strict monotonicity of $\Phi$ identifies every $\lambda_j$. Since $\sigma$ is already identified, all components
\[
\mu_j=\sigma\lambda_j
\]
are uniquely determined. This proves global identifiability.
\end{proof}

\begin{corollary}[Corollary of Theorem~\ref{thm:ident}]
\label{cor:alpha-sigma}
For
\[
Z\sim\N_{n+1}(\alpha e_1,\sigma^2I_{n+1}),
\quad n\ge2,
\]
the pair $(\alpha,\sigma)$ is globally identifiable. In particular,
\[
\frac{\alpha}{\sigma}
=
\Phi^{-1}\{\mathbb P(\cos\Theta_n>0)\},
\]
and $\sigma$ is then recovered uniquely from \eqref{eq:sigma-identification}.
\end{corollary}

\begin{remark}
For $n=1$, the observation is only the direction of a bivariate isotropic Gaussian. Then $(\mu,\sigma)$ and $(c\mu,c\sigma)$, $c>0$, induce the same angular law, so absolute scale is not identifiable. The outer shift $R_n$ entering the event \eqref{eq:scale-event} is exactly what breaks this scale invariance when $n\ge2$.
\end{remark}

\begin{remark}
Theorem~\ref{thm:ident} shows that the non-monotone finite-sample RMSE observed for $\sigma$ in the preliminary $n=3$ simulation cannot be attributed to structural non-identifiability of the isotropic model. It should instead be investigated as a finite-sample, optimization, or numerical-integration effect. This distinction is important when interpreting the simulation study.
\end{remark}

\subsection{Basics of maximum log-likelihood theory}
\label{sec:basics_max_log_likelihood}

Let
\begin{equation}
\vartheta=(\mu^\top,\sigma)^\top\in\R^{n+2},
\quad
Z\sim\mathcal N_{n+1}(\mu,\sigma^2 I_{n+1}),
\label{eq:param-vartheta}
\end{equation}
and let {$p_{\pr(Z)}(\bm{\theta})$} denote the exact angular density in
Theorem~\ref{thm:pgdensity}.  Because the surface Jacobian in
\eqref{eq:fiber-jac} is independent of $\vartheta$, maximizing the angular
likelihood is equivalent to maximizing the likelihood with respect to
Riemannian surface volume.

For  statistical estimation of parameters of type \eqref{eq:param-vartheta}, we work on the compact parameter set 
\begin{equation}
\Xi=K_\mu\times[\underline\sigma,\overline\sigma],
\quad
0<\underline\sigma<\overline\sigma<\infty,
\label{eq:compact-parameter-space}
\end{equation}
where $K_\mu\subset\R^{n+1}$ is compact { such that $F_n(\T_n)\subseteq \operatorname{int}K_\mu$} and the true parameter
$\vartheta_0=(\mu_0^\top,\sigma_0)^\top$ belongs to the interior of $\Xi$ {under the product topology}.
For independent observations $\Theta_1,\ldots,\Theta_N\sim { p_{\vartheta_0} (\bm{\theta})} $ {  where $p_{\vartheta_0}$ denotes the probability density function of $p_{\pr(Z)} $ where $Z\sim N_{n+1}(\mu_0,\sigma_0^2 I)$,}
define
\begin{equation}
\ell_N(\vartheta)
=
\sum_{i=1}^N \log { p_{\vartheta} }(\Theta_i),
\quad
\widehat\vartheta_N
\in
\operatorname*{arg\,max}_{\vartheta\in\Xi}\ell_N(\vartheta).
\label{eq:mle-def}
\end{equation}

\subsection{A regularity lemma for the projected normal}
\label{sec:regularitylemmaforprojnormal}

In this section we shall introduce a lemma treating the regularity of the projected normal density function given in \eqref{eq:angular-density}.

\begin{lemma}
\label{lem:likelihood-regularity}
There exists a compact parameter space
\eqref{eq:compact-parameter-space} such that the following hold.
\begin{enumerate}
\item There are constants $0<c<C<\infty$ such that
\begin{equation}
c\le {p_\vartheta (\bm{\theta})}\le C,
\quad
({\bm{\theta}},\vartheta)\in\T^n\times\Xi.
\label{eq:density-bounds}
\end{equation}
\item For each fixed ${\bm{\theta}}$, ${p_\vartheta (\bm{\theta})}$ is infinitely
differentiable in $\vartheta$ on the interior of $\Xi$.  Derivatives of
any fixed finite order may be passed under the $T$-integral.
\item { All orders of} derivatives of $\log g_\vartheta({\bm{\theta}})$ with
respect to $\vartheta$ are uniformly bounded on every compact subset of
$\operatorname{int}(\Xi)$, uniformly in ${\bm{\theta}}\in\T^n$.
\end{enumerate}
\end{lemma}

\begin{proof}
{ At the curved torus $F_n(\T_n)$,} the strong-radius condition gives the uniform lower
bounds
\[
q_1(R_1,{\bm{\theta}})=R_1,
\quad
q_k(R_1,{\bm{\theta}})
\ge
R_k-\sum_{j<k}R_j
=:d_k>0.
\]
 {By \eqref{eq:q-affine} $|\partial_T q_k|=|\prod_{j=1}^{k-1}\cos\theta_j|\le1$}, there exists
$\varepsilon>0$, independent of ${\bm{\theta}}$, such that
\[
[R_1-\varepsilon,R_1+\varepsilon]
\subset (0,T_0({\bm{\theta}}))
\quad\text{for every }{\bm{\theta}}\in\T^n.
\]
On this interval, $X(T,{\bm{\theta}})$ and $J(T,{\bm{\theta}})$ are
continuous and $J$ is bounded away from zero.  Since $K_\mu$ is compact
and $\sigma$ is bounded above and away from zero, the Gaussian density is
also bounded below by a positive constant.  Integrating the $T$-coordinate over $[R_1-\varepsilon,R_1+\varepsilon]$
 proves the lower bound in \eqref{eq:density-bounds}.
 
The upper bound follows from observing that integrated expression in \eqref{eq:density-moment-reduction} is bounded from above.

{ Note that in the case that the noise is isotropic it holds that $A= \frac{1}{\sigma^2}$ is constant in ${\bm{\theta}}$. Moreover $B(\bm{\theta}) = \frac{1}{\sigma^2}(\langle F_n(\bm{\theta}),\nu(\bm{\theta})\rangle -R_1 -\langle\nu(\bm{\theta}) ,\mu\rangle)$ and therefore $K(\bm{\theta})=\frac{1}{\sigma}(\langle F_n(\bm{\theta}),\nu(\bm{\theta})\rangle -R_1 -\langle\nu(\bm{\theta}) ,\mu\rangle)$, since $\sigma$ is bounded away from zero and since $p_{\vartheta}$ is a smooth function of $\sigma$ and $\mu$, it holds that $p_{\vartheta}$ is infinitely differentiable in $\vartheta$. Furthermore, the domain of integration for integration over $T$ is the same for all $\vartheta\in \Xi$, then boundedness of the derivatives ensures that the derivatives can be moved inside the integral. Since $p_{\vartheta}$ is bounded away from $0$ the map $\log(p_{\vartheta})$ is infinitely differentiable as well and they are bounded in $\operatorname{int}(\Xi)$.
}
\end{proof}

\subsection{Proof of Theorem~\ref{thm:mle}}
\label{sec:proof_mle}

For $\vartheta$ as defined in \eqref{eq:param-vartheta}, let
$$
s_\vartheta(\bm{\theta})
=
\nabla_\vartheta \log p_\vartheta(\bm{\theta})
$$
denote the score function, and define the Fisher information matrix by
$$
\mathcal I(\vartheta_0)
=
\int_{\T^n}
s_{\vartheta_0}(\bm{\theta})
s_{\vartheta_0}(\bm{\theta})^\top
p_{\vartheta_0}(\bm{\theta})
\dd \bm{\theta}.
$$

\begin{proof}[Proof of Theorem~\ref{thm:mle}]

We first prove the strong consistency in
\eqref{eq:mle-strong-consistency}.  Define the normalized empirical
log-likelihood
$$
M_N(\vartheta)
=
\frac{1}{N}\ell_N(\vartheta)
=
\frac{1}{N}
\sum_{i=1}^N
\log p_\vartheta(\Theta_i),
$$
and its population counterpart
$$
M(\vartheta)
=
\mathbb E_{\vartheta_0}
\{\log p_\vartheta(\Theta)\}
=
\int_{\T^n}
\log p_\vartheta(\bm{\theta})\,
p_{\vartheta_0}(\bm{\theta})
\dd\bm{\theta}.
$$

By Lemma~\ref{lem:likelihood-regularity}, there exist constants
$0<c<C<\infty$ such that
$$
c
\le
p_\vartheta(\bm{\theta})
\le
C,
\quad
(\bm{\theta},\vartheta)
\in
\T^n\times\Xi.
$$
Consequently,
$$
|\log p_\vartheta(\bm{\theta})|
\le
\max\{|\log c|,|\log C|\},
$$
uniformly in $(\bm{\theta},\vartheta)$, and hence
$$
\int_{\T^n}
|\log p_\vartheta(\bm{\theta})|
p_{\vartheta_0}(\bm{\theta})
\dd\bm{\theta}
<\infty.
$$

For every fixed $\vartheta\in\Xi$, the strong law of large numbers
therefore gives
$$
M_N(\vartheta)
\longrightarrow
M(\vartheta)
\quad\text{almost surely}.
$$

We now strengthen this pointwise convergence to uniform convergence
over $\Xi$.  Since $\T^n\times\Xi$ is compact and
$(\bm{\theta},\vartheta)\mapsto
\log p_\vartheta(\bm{\theta})$
is continuous, it is uniformly continuous.  Define its modulus of
continuity in the parameter argument by
$$
\omega(\delta)
=
\sup_{\substack{
\bm{\theta}\in\T^n,\\
\vartheta,\vartheta'\in\Xi,\\
\|\vartheta-\vartheta'\|\le\delta
}}
\left|
\log p_\vartheta(\bm{\theta})
-
\log p_{\vartheta'}(\bm{\theta})
\right|.
$$
Then
$$
\omega(\delta)\longrightarrow0
\quad\text{as }\delta\downarrow0.
$$

Fix $\delta>0$.  Since $\Xi$ is compact, there exists a finite
$\delta$-net
$
\vartheta^{(1)},\ldots,\vartheta^{(K)}
$
such that for every $\vartheta\in\Xi$ there is some $k$ satisfying
$
\|\vartheta-\vartheta^{(k)}\|\le\delta.
$
For such a pair,
$$
|M_N(\vartheta)-M_N(\vartheta^{(k)})|
\le
\omega(\delta),
$$
because each summand differs by at most $\omega(\delta)$.  Similarly,
$$
|M(\vartheta)-M(\vartheta^{(k)})|
\le
\omega(\delta).
$$
Hence
\begin{multline*}
|M_N(\vartheta)-M(\vartheta)|\le
|M_N(\vartheta)-M_N(\vartheta^{(k)})|+
|M_N(\vartheta^{(k)})-M(\vartheta^{(k)})|+
|M(\vartheta^{(k)})-M(\vartheta)|
\\
\le
2\omega(\delta)
+
|M_N(\vartheta^{(k)})-M(\vartheta^{(k)})|.
\end{multline*}
Taking the supremum over $\vartheta\in\Xi$ gives
$$
\sup_{\vartheta\in\Xi}
|M_N(\vartheta)-M(\vartheta)|
\le
2\omega(\delta)
+
\max_{1\le k\le K}
|M_N(\vartheta^{(k)})-M(\vartheta^{(k)})|.
$$
The net contains only finitely many points, so the strong law applied
simultaneously to these $K$ points implies
$$
\max_{1\le k\le K}
|M_N(\vartheta^{(k)})-M(\vartheta^{(k)})|
\longrightarrow0
\quad\text{almost surely}.
$$
Therefore,
$$
\limsup_{N\to\infty}
\sup_{\vartheta\in\Xi}
|M_N(\vartheta)-M(\vartheta)|
\le
2\omega(\delta)
\quad\text{almost surely}.
$$
Since $\delta>0$ is arbitrary and
$\omega(\delta)\to0$, we obtain
\begin{equation}
\sup_{\vartheta\in\Xi}
\left|
\frac{1}{N}\ell_N(\vartheta)-M(\vartheta)
\right|
\longrightarrow0
\quad\text{almost surely}.
\label{eq:uniform-slln-likelihood}
\end{equation}

We next identify the unique maximizer of the population criterion.
For every $\vartheta\in\Xi$,
\begin{align*}
M(\vartheta)-M(\vartheta_0)
&=
\int_{\T^n}
p_{\vartheta_0}(\bm{\theta})
\log
\left\{
\frac{p_\vartheta(\bm{\theta})}
     {p_{\vartheta_0}(\bm{\theta})}
\right\}
\dd\bm{\theta}
\\
&=
-
\operatorname{KL}
\{p_{\vartheta_0}\|p_\vartheta\}
\\
&\le0.
\end{align*}
Thus $\vartheta_0$ is a maximizer of $M$.  Equality holds if and only
if
$
p_\vartheta
=
p_{\vartheta_0}
\quad\text{almost everywhere}.
$
By Theorem~\ref{thm:ident}, the isotropic projected-normal family is
identifiable, and hence this equality implies
$
\vartheta=\vartheta_0.
$
Therefore, $\vartheta_0$ is the unique maximizer of $M$.

It remains to show that the empirical maximizer converges to this
unique population maximizer.  Fix $\varepsilon>0$ and define
$$
A_\varepsilon
=
\left\{
\vartheta\in\Xi:
\|\vartheta-\vartheta_0\|\ge\varepsilon
\right\}.
$$
Since $\Xi$ is compact, $A_\varepsilon$ is compact.  Since $M$ is
continuous and has the unique maximizer $\vartheta_0$,
$$
\Delta_\varepsilon
:=
M(\vartheta_0)
-
\sup_{\vartheta\in A_\varepsilon}
M(\vartheta)
>0.
$$

By the uniform convergence in
\eqref{eq:uniform-slln-likelihood}, almost surely there exists
$N_\varepsilon<\infty$ such that, for all
$N\ge N_\varepsilon$,
$$
\sup_{\vartheta\in\Xi}
|M_N(\vartheta)-M(\vartheta)|
<
\frac{\Delta_\varepsilon}{3}.
$$
Hence
$$
M_N(\vartheta_0)
>
M(\vartheta_0)
-
\frac{\Delta_\varepsilon}{3}.
$$
On the other hand, for every
$\vartheta\in A_\varepsilon$,
$$
M_N(\vartheta)
<
M(\vartheta)
+
\frac{\Delta_\varepsilon}{3}
\le
M(\vartheta_0)
-
\frac{2\Delta_\varepsilon}{3}.
$$
Therefore,
$$
M_N(\vartheta_0)
>
\sup_{\vartheta\in A_\varepsilon}
M_N(\vartheta)
$$
for all sufficiently large $N$, almost surely.

Since
$$
\widehat\vartheta_N
\in
\arg\max_{\vartheta\in\Xi}
M_N(\vartheta),
$$
it follows that
$
\widehat\vartheta_N
\notin
A_\varepsilon
$
for all sufficiently large $N$, almost surely.  Equivalently,
$$
\|\widehat\vartheta_N-\vartheta_0\|
<
\varepsilon
$$
eventually almost surely.  Since $\varepsilon>0$ was arbitrary,
$$
\widehat\vartheta_N
\longrightarrow
\vartheta_0
\quad\text{almost surely}.
$$
This proves \eqref{eq:mle-strong-consistency}.

Next, we shall show that the convergence in distribution \eqref{eq:mle-clt} holds true. {By Lemma~\ref{lem:likelihood-regularity}, $\log p_\vartheta$ has smooth regularity, since $p_\vartheta$ has the same integration domain for all $\vartheta \in \Xi$ }
\[
\mathbb E_{\vartheta_0}\{s_{\vartheta_0}(\Theta)\}=0,
\quad
-\mathbb E_{\vartheta_0}
\{\nabla_\vartheta^2\log {p}_{\vartheta_0}(\Theta)\}
=
\mathcal I(\vartheta_0).
\]
The  central limit theorem gives
\begin{equation}
\frac1{\sqrt N}
\sum_{i=1}^N s_{\vartheta_0}(\Theta_i)
\xrightarrow{d}
\mathcal N_{n+2}\{0,\mathcal I(\vartheta_0)\}.
\end{equation}
By \eqref{eq:mle-strong-consistency}, $\widehat\vartheta_N$ lies in the
interior with probability tending to one.  Therefore, for large enough, $N$ it holds that $\hat{\vartheta}_N\in \operatorname{int}\Xi$ and therefore conditional

On the event
$
E_N=\{\widehat\vartheta_N\in\operatorname{int}\Xi\},
$
the first-order condition gives
$$
\nabla_\vartheta\ell_N(\widehat\vartheta_N)=0.
$$
By a Taylor expansion around $\vartheta_0$,
$$
0
=
\nabla_\vartheta\ell_N(\vartheta_0)
+
\nabla_\vartheta^2\ell_N(\widetilde\vartheta_N)
(\widehat\vartheta_N-\vartheta_0),
$$
for an intermediate parameter $\widetilde\vartheta_N$ between
$\vartheta_0$ and $\widehat\vartheta_N$. Hence,
$$
0
=
\frac1{\sqrt N}\sum_{i=1}^N
s_{\vartheta_0}(\Theta_i)
+
\left\{
\frac1N
\nabla_\vartheta^2\ell_N(\widetilde\vartheta_N)
\right\}
\sqrt N(\widehat\vartheta_N-\vartheta_0),
$$
on the event $E_N$.

Since $\widetilde\vartheta_N$ lies on the line segment joining
$\vartheta_0$ and $\widehat\vartheta_N$, the strong consistency
\eqref{eq:mle-strong-consistency} implies
$$
\widetilde\vartheta_N\longrightarrow\vartheta_0
\quad\text{almost surely}.
$$
By Lemma~\ref{lem:likelihood-regularity}, the Hessian is continuous
and uniformly bounded on a neighborhood of $\vartheta_0$.  Therefore,
by the corresponding uniform law of large numbers,
$$
\frac1N
\nabla_\vartheta^2
\ell_N(\widetilde\vartheta_N)
\longrightarrow
\mathbb E_{\vartheta_0}
\left\{
\nabla_\vartheta^2
\log p_{\vartheta_0}(\Theta)
\right\}
=
-\mathcal I(\vartheta_0)
$$
almost surely.
Thus
$$
0
=
\frac1{\sqrt N}\sum_{i=1}^N
s_{\vartheta_0}(\Theta_i)
-
\mathcal I(\vartheta_0)
\sqrt N(\widehat\vartheta_N-\vartheta_0)
+
o_p(1).
$$
Since $\mathcal I(\vartheta_0)$ is positive definite and hence
invertible, and since $\mathbb P(E_N)\to1$, Slutsky's theorem together
with the score central limit theorem yields
\eqref{eq:mle-clt}.
\end{proof}

\subsection{Consistency of the simulation-based minimum-MMD estimator}
\label{app:flexible-mmd}
Let $\{P_\eta:\eta\in\mathcal H\}$ be a projected normal family and let
\[
J(\eta)=D_k^2(P_0,P_\eta)
\]
denote the population squared MMD after the sine-cosine embedding.  Let
$\widehat J_{N,m_N}(\eta)$ be the corresponding simulated empirical criterion.



\begin{proof}[Proof of Proposition~\ref{prop:mmd}]
Continuity of $J$ on the compact set $\mathcal H$ implies that the
minimizer set
$$
\mathcal H_0
=
\arg\min_{\eta\in\mathcal H}J(\eta)
$$
is nonempty and compact.  Fix $\varepsilon>0$ and define
$$
\mathcal H_\varepsilon
=
\left\{
\eta\in\mathcal H:
\operatorname{dist}(\eta,\mathcal H_0)\ge\varepsilon
\right\}.
$$
Since $\mathcal H_\varepsilon$ is compact and disjoint from
$\mathcal H_0$,
$$
\Delta_\varepsilon
=
\inf_{\eta\in\mathcal H_\varepsilon}J(\eta)
-
\inf_{\eta\in\mathcal H}J(\eta)
>0.
$$

Suppose
$
\sup_{\eta\in\mathcal H}
\left|
\widehat J_{N,m_N}(\eta)-J(\eta)
\right|
<
\frac{\Delta_\varepsilon}{3}
$

and
$
r_N<\frac{\Delta_\varepsilon}{3},
$
where
$$
\widehat J_{N,m_N}(\widehat\eta_N)
\le
\inf_{\eta\in\mathcal H}
\widehat J_{N,m_N}(\eta)+r_N.
$$
For any $\eta\in\mathcal H_\varepsilon$,
$$
\widehat J_{N,m_N}(\eta)
\ge
\inf_{\eta'\in\mathcal H}J(\eta')
+
\frac{2\Delta_\varepsilon}{3},
$$
whereas for any $\eta_0\in\mathcal H_0$,
$$
\widehat J_{N,m_N}(\eta_0)
\le
\inf_{\eta'\in\mathcal H}J(\eta')
+
\frac{\Delta_\varepsilon}{3}.
$$
Hence every empirical criterion value on
$\mathcal H_\varepsilon$ exceeds the value at a population minimizer
by at least $\Delta_\varepsilon/3$, which is larger than the allowed
optimization error.  Therefore
$
\widehat\eta_N\notin\mathcal H_\varepsilon.
$
Since the displayed event has probability tending to one,
$$
\operatorname{dist}(\widehat\eta_N,\mathcal H_0)
\xrightarrow{p}0.
$$

If the model is correctly specified and the kernel is characteristic,
then
$$
J(\eta)=0
\iff
P_\eta=P_{\eta_0}.
$$
Hence the preceding set convergence and continuity of $J$ imply
$$
D_k(P_{\widehat\eta_N},P_{\eta_0})
\xrightarrow{p}0.
$$
If the generating parameter is unique, so that
$\mathcal H_0=\{\eta_0\}$, then
$
\widehat\eta_N\xrightarrow{p}\eta_0.
$
\end{proof}
\begin{remark}
The distinction between distributional and parameter consistency is material for the
diagonal-plus-rank-one covariance used in the applications because $vv^\top$ is invariant
under $v\mapsto-v$.  The proposition therefore does not require an artificial sign
constraint merely to state the distributional guarantee.
\end{remark}

\begin{remark}
The proposition is an asymptotic statement for an idealized sequence of simulation-based minimum-MMD estimators satisfying $m_N\to\infty$, uniform convergence of the simulated criterion, and asymptotically negligible optimization error. The finite-batch stochastic optimizer used in the experiments is a computational approximation to this estimator. We do not claim that the fixed finite-batch optimization scheme itself satisfies the asymptotic conditions of the proposition.
\end{remark}

For the flexible model used in the applications, the recursively curved
geometry and projection map are fixed before estimation.  Only the
ambient Gaussian parameters are estimated,
$$
Z\sim N_{n+1}(\mu,\Sigma),
$$
with
\begin{equation}
\Sigma=
\operatorname{diag}
\left(
\lambda_{\min}+e^{2s_1},\ldots,
\lambda_{\min}+e^{2s_{n+1}}
\right)
+vv^\top,
\label{eq:flexcov-app}
\end{equation}
where $\lambda_{\min}>0$ is fixed and
$s,v\in\mathbb R^{n+1}$.

The diagonal term controls coordinate-specific ambient variation,
whereas $vv^\top$ introduces a parsimonious dependence structure.
In particular,
$$
\Sigma_{jj}
=
\lambda_{\min}+e^{2s_j}+v_j^2,
\quad
\Sigma_{jk}=v_jv_k,\quad j\ne k.
$$
Since the diagonal matrix in \eqref{eq:flexcov-app} is strictly
positive definite and $vv^\top$ is positive semidefinite,
$\Sigma\succ0$ throughout optimization.

The covariance parameterization uses $2(n+1)$ optimization
parameters through $(s,v)$, compared with
$
\frac{(n+1)(n+2)}{2}
$
parameters for an unrestricted covariance matrix.  Including
$\mu\in\mathbb R^{n+1}$ gives $3(n+1)$ parameters for M4.

For $\eta=(\mu,s,v)$, let $P_\eta$ denote the projected angular law.
We fit the flexible family by simulation-based minimum MMD after the
periodic embedding
$$
\Phi(\theta)
=
(\cos\theta_1,\ldots,\cos\theta_n,
 \sin\theta_1,\ldots,\sin\theta_n)
\in\mathbb R^{2n}.
$$
The simulated empirical criterion compares the embedded observed
sample with embedded draws from $P_\eta$ using a multiscale kernel.
The parameterization is not one-to-one because
$$
vv^\top=(-v)(-v)^\top.
$$
Thus $v$ and $-v$ define the same covariance and the same projected
distribution.  This motivates the set-valued and distributional form
of Proposition~\ref{prop:mmd}.

In the applications, the idealized criterion is approximated using
finite simulated mini-batches and stochastic optimization.  Each
iteration uses an observed mini-batch and a simulated mini-batch from
the current model, evaluates the multiscale squared-MMD objective, and
updates $\eta$ by stochastic gradient optimization.  Restarts and
tuning are selected using validation data only.  Proposition~\ref{prop:mmd}
therefore describes the asymptotic target of the procedure and should
not be interpreted as a finite-sample convergence guarantee for the
particular stochastic optimization run.

\subsection{Energy-based Fr\'echet statistics on the curved torus}
\label{sec:energy-frechet}

This section formulates the intrinsic statistics directly through the
Riemannian energy functional.  
 { We work with the energy functionals since closed form geodesics are not easily attainable.}

For { a smooth curve $\gamma:[0,1]\rightarrow \R^n$}, define its energy { on the embedded $n$-torus $F_n(\T^n)$} by
\[
\mathcal E(\gamma)
=
\int_0^1
\sum_{j=1}^n
\rho_j\{\gamma(t)\}^2\dot\gamma_j(t)^2\,\dd t.
\]

Let $p,y\in\T^n$, and let $\widetilde p,\widetilde y\in\R^n$
be any chosen lifts of these points.  Here, for
$x\in\R^n$, we write
$$
[x]
=
x+(2\pi\mathbb Z)^n
\in
\T^n
=
\R^n/(2\pi\mathbb Z)^n
$$
for the equivalence class of $x$ modulo $2\pi$ in each coordinate.
Thus
$
[\widetilde p]=p,
\quad
[\widetilde y]=y.
$  For $k\in\mathbb Z^n$, let
\[
    \mathcal A_k(p,y)
=
\big\{
\gamma\in C^{{\infty}}([0,1],\R^n):
\gamma(0)=\widetilde p,\quad \gamma(1)=\widetilde y+2\pi k
\big\}.
\]
Define
\begin{equation}
D_E^2(p,y)
=
\min_{k\in\mathbb Z^n}
\inf_{\gamma\in\mathcal A_k(p,y)}
\mathcal E(\gamma).
\label{eq:energy-distance}
\end{equation}

Set
\[
\underline\rho
=
\min\left\{
R_1,
\min_{2\le j\le n}
\left(R_j-\sum_{\ell<j}R_\ell\right)
\right\}>0,
\]
and
\[
\overline\rho
=
\max\left\{
R_1,
\max_{2\le j\le n}
\left(R_j+\sum_{\ell<j}R_\ell\right)
\right\}<\infty.
\]
Then, for every lifted  curve,
\begin{equation}
\underline\rho^2
\int_0^1\|\dot\gamma(t)\|^2\dd t
\le
\mathcal E(\gamma)
\le
\overline\rho^2
\int_0^1\|\dot\gamma(t)\|^2\dd t.
\label{eq:energy-sandwich}
\end{equation}

\begin{proof}[Proof of Proposition~\ref{prop:winding-attainment}]

For a fixed winding number $k\in\mathbb Z^n$, every
$\gamma\in\mathcal A_k(p,y)$ satisfies, by Cauchy-Schwarz and
\eqref{eq:energy-sandwich},
\begin{equation}
\mathcal E(\gamma)
\ge
\underline\rho^2
\int_0^1 \|\dot\gamma(t)\|^2\,\dd t
\nonumber
\ge
\underline\rho^2
\left\|
\widetilde y+2\pi k-\widetilde p
\right\|^2.
\label{eq:winding-lower-bound}
\end{equation}

On the other hand, choose any fixed reference winding
$k_0\in\mathbb Z^n$.  The straight lifted curve
$$
\gamma_0(t)
=
\widetilde p
+
t\{\widetilde y+2\pi k_0-\widetilde p\},
\qquad 0\le t\le1,
$$
has finite energy; denote it by $U_0<\infty$.  Therefore, any
$k\in\mathbb Z^n$ satisfying
$$
\underline\rho^2
\left\|
\widetilde y+2\pi k-\widetilde p
\right\|^2
>
U_0
$$
cannot attain the minimum in \eqref{eq:energy-distance}.
Since
$$
\left\|
\widetilde y+2\pi k-\widetilde p
\right\|
\longrightarrow\infty
\qquad\text{as }\|k\|\to\infty,
$$
only finitely many winding numbers can be relevant.

Now consider the lifted Riemannian metric $\widetilde g$ on
$\mathbb R^n$ obtained by periodically extending the metric on
$\T^n$.  By \eqref{eq:energy-sandwich},
$$
\underline\rho^2\|v\|^2
\le
\widetilde g_x(v,v)
\le
\overline\rho^2\|v\|^2,
\qquad
x,v\in\mathbb R^n.
$$
Hence the lifted Riemannian metric is uniformly equivalent to the
Euclidean metric.  In particular, $(\mathbb R^n,\widetilde g)$ is
complete.

Therefore, by the Hopf-Rinow theorem, for every fixed
$k\in\mathbb Z^n$ there exists a minimizing geodesic joining
$\widetilde p$ to $\widetilde y+2\pi k$.  Let
$\widetilde d_g$ denote the geodesic distance associated with the
lifted metric.  After constant-speed parametrization, the minimizing\\
geodesic $\gamma_k\in\mathcal A_k(p,y)$ satisfies
$$
\mathcal E(\gamma_k)
=
\widetilde d_g
\left(
\widetilde p,\widetilde y+2\pi k
\right)^2.
$$

Since only finitely many winding classes can attain the global
minimum, there exists some $k^\star\in\mathbb Z^n$ and a corresponding
constant-speed minimizing geodesic $\gamma^\star$ such that
$$
D_E^2(p,y)
=
\mathcal E(\gamma^\star)
=
\min_{k\in\mathbb Z^n}
\widetilde d_g
\left(
\widetilde p,\widetilde y+2\pi k
\right)^2.
$$

Because the metric on $\mathbb R^n$ is the periodic lift of the metric
on $\T^n$, the geodesic distance on the quotient torus satisfies
$$
d_g(p,y)
=
\min_{k\in\mathbb Z^n}
\widetilde d_g
\left(
\widetilde p,\widetilde y+2\pi k
\right).
$$
Finally, $F_n$ is an isometric embedding of $(\T^n,g)$ into
$F_n(\T^n)$ with its induced Riemannian metric.  Hence
$$
D_E^2(p,y)
=
d_g(p,y)^2
=
d_M\{F_n(p),F_n(y)\}^2.
$$
This proves \eqref{eq:distance}.
\end{proof}

Let $Y$ be a random  { variable} on $F_n(\T^n)$.  Define the energy-based
Fr\'echet function
\[
\mathcal F(p)
=
\mathbb E\{D_E^2({F_n^{-1}(p)},{F_n^{-1}(Y)})\},
\quad p\in F_n(\T^n),
\]
and the population mean set
\[
\mathcal M
=
\operatorname*{arg\,min}_{p\in F_n(\T^n)}\mathcal F(p).
\]
For observations $Y_1,\ldots,Y_N$, define the respective empirical counterparts
\begin{equation}
    \begin{aligned}
        \mathcal F_N(p)
&=
\frac1N\sum_{i=1}^N D_E^2({F_n^{-1}(p)},{F_n^{-1}(Y_i)}),\\
\widehat m_N
&\in
\operatorname*{arg\,min}_{p\in F_n(\T^n)}\mathcal F_N(p).
    \end{aligned}
    \label{eq:empirical-frechet}
\end{equation}

\begin{theorem}
\label{thm:frechet-existence}
For every probability distribution of $Y$ on $F(\T^n)$, the population
mean set $\mathcal M$ is nonempty and compact.  For every finite sample,
 { $\hat{m}_N$ is also nonempty and compact}.
\end{theorem}

\begin{proof}
The embedded torus is compact.  By
Proposition~\ref{prop:winding-attainment}, $D_E=d_g$, and the Riemannian
distance is continuous.  Since the manifold has finite diameter,
$D_E^2(p,Y)$ is bounded uniformly in $(p,Y)$.  Hence
$\mathcal F$ is continuous by dominated convergence and attains its
minimum on the compact manifold.  The same argument applies directly to
$\mathcal F_N$.
\end{proof}

Global uniqueness cannot be guaranteed for arbitrary distributions on a
torus.  A sufficient local condition is the following.

\begin{proposition}
\label{prop:frechet-uniqueness}
Suppose there is a geodesically convex open set $U$ containing the support
of $Y$ and a constant $\lambda>0$ such that, for every $y$ in the support
and every $p\in U$,
\begin{equation}
\operatorname{Hess}_p
\left\{\frac12D_E^2(p,y)\right\}
\succeq
\lambda g_p.
\label{eq:strong-convexity-distance}
\end{equation}
Then $\mathcal F$ is strictly geodesically convex on $U$ and has at most
one minimizer there.  If, additionally,
\begin{equation}
\inf_{p\notin U}\mathcal F(p)
>
\inf_{p\in U}\mathcal F(p),
\label{eq:no-outside-minimizer}
\end{equation}
then the population Fr\'echet mean is globally unique.
\end{proposition}

\begin{proof}
Taking expectations in \eqref{eq:strong-convexity-distance} preserves the
Hessian lower bound, so $\mathcal F$ is $2\lambda$-strongly geodesically
convex on $U$.  Therefore it has at most one minimizer in $U$.
Condition \eqref{eq:no-outside-minimizer} excludes all minimizers outside
$U$.
\end{proof}

\begin{theorem}
\label{thm:frechet-consistency}
Assume the population Fr\'echet mean is unique, say
$\mathcal M=\{m_F\}$.  Then any empirical minimizer in
\eqref{eq:empirical-frechet} satisfies
\[
D_E(\widehat m_N,m_F)\longrightarrow0
\quad\text{almost surely}.
\]
\end{theorem}

\begin{proof}

{ Since $F_n(\T^n)$ is compact this is a known result from \cite{bhattacharya2003large}.}
\end{proof}

The diagonal metric provides a global orthonormal frame
\begin{equation}
e_j(p)
=
\rho_j(p)^{-1}\frac{\partial}{\partial\theta_j},
\quad j=1,\ldots,n.
\label{eq:global-orthonormal-frame}
\end{equation}
Suppose that for $(p,y)$ in a neighborhood of
$\{m_F\}\times\operatorname{supp}(Y)$ there is a unique minimizing
constant-speed path for \eqref{eq:energy-distance}, denoted
$\gamma_{p,y}$.  Define its initial-velocity coordinates in the frame
\eqref{eq:global-orthonormal-frame} by
\[
V_E(p,y)
=
\begin{pmatrix}
\rho_1(p)\dot\eta_{p,y,1}(0)\\
\vdots\\
\rho_n(p)\dot\eta_{p,y,n}(0)
\end{pmatrix}
\in\R^n.
\]
Because the minimizing path has constant speed,
\begin{equation}
\|V_E(p,y)\|^2
=
\mathcal E(\gamma_{p,y})
=
D_E^2(p,y).
\label{eq:velocity-energy-identity}
\end{equation}
This is the energy-functional version of the intrinsic tangent vector.

Assume the unique minimizers above depend continuously on $(p,y)$ in
initial velocity.  This holds, for example, when the relevant pairs avoid
the cut locus and remain inside a compact strongly convex normal
neighborhood.  Under this condition, differentiating the Fr\'echet
function along tangent directions gives
\[
\mathbb E\{V_E(m_F,Y)\}=0.
\]
We therefore define the intrinsic covariance by
\[
\Sigma_F
=
\mathbb E
\left[
V_E(m_F,Y)V_E(m_F,Y)^\top
\right].
\]
Its empirical counterpart is
\[
\widehat\Sigma_{F,N}
=
\frac1N\sum_{i=1}^N
V_E(\widehat m_N,Y_i)
V_E(\widehat m_N,Y_i)^\top.
\]
Operationally, $V_E(\widehat m_N,Y_i)$ is obtained from the initial
velocity of the same energy-minimizing path computation already used to
evaluate $D_E^2(\widehat m_N,Y_i)$.

\begin{theorem}
\label{thm:covariance-consistency}
Assume the hypotheses of Theorem~\ref{thm:frechet-consistency}.  Suppose,
in addition, that there exists a neighborhood $U$ of $m_F$ such that for
all $p\in U$ and all $y$ in the support of $Y$ the minimizing
energy path is unique and the map
\begin{equation}
(p,y)\longmapsto V_E(p,y)
\end{equation}
is continuous.  Then
\begin{equation}
\widehat\Sigma_{F,N}
\longrightarrow
\Sigma_F
\quad\text{almost surely}
\label{eq:covariance-strong-consistency}
\end{equation}
in any matrix norm.
\end{theorem}

\begin{proof}
Because the torus and the relevant support set are compact and
$V_E$ is continuous, $V_E$ is uniformly bounded on
$U\times\operatorname{supp}(Y)$.  By
Theorem~\ref{thm:frechet-consistency}, $\widehat m_N\to m_F$ almost
surely; hence uniform continuity gives
\[
\sup_{y\in\operatorname{supp}(Y)}
\big\|
V_E(\widehat m_N,y)V_E(\widehat m_N,y)^\top
-
V_E(m_F,y)V_E(m_F,y)^\top
\big\|
\longrightarrow0
\]
almost surely.  Therefore
\[
\left\|
\widehat\Sigma_{F,N}
-
\frac1N\sum_{i=1}^N
V_E(m_F,Y_i)V_E(m_F,Y_i)^\top
\right\|\longrightarrow 0    
\]
almost surely. The remaining empirical average converges almost surely to
$\Sigma_F$ by the strong law of large numbers, proving
\eqref{eq:covariance-strong-consistency}.
\end{proof}

\begin{remark}
The definition \eqref{eq:energy-distance} is exact and minimizes over all
winding classes.  The numerical algorithm used in the simulation study
approximates the path problem by a finite mesh and a finite candidate
winding set.  The mesh-refinement and first-order stationarity results
validate that implementation empirically for the tested cases, but they
are not a proof that a particular finite numerical search finds the
global minimizer for every possible pair of points.
\end{remark}

\section{Additional numerical validation}
\label{app:simulation-validation}

This appendix reports the detailed numerical results underlying
Section~\ref{sec:numerical-validation}.  We use the
same recursively separated radii as in the main experiments,
$R_1=1$ and $R_k=1.05\sum_{j<k}R_j$, and the dimension grid
$n\in\{2,3,5,10\}$ unless otherwise stated.  The purpose of these
experiments is numerical validation of the geometric, probabilistic,
and optimization components; they do not constitute as a large
parameter-grid simulation study.

\subsection{Volume law, exact sampler, and normal-fiber coordinates}
\label{app:uniform-sampler}
The total Riemannian volume is $\operatorname{Vol}\{F_n(\T^n)\}=(2\pi)^n\prod_{j=1}^nR_j$. We checked this expression by Monte Carlo integration using $10^5$
independent angular points. The exact sequential surface-uniform
sampler was assessed using $20{,}000$ draws and the Rosenblatt
probability-integral transform (PIT).
Table~\ref{tab:uniform-audit} summarizes the resulting volume and
PIT diagnostics across the four investigated dimensions.

\begin{table}[ht]
\centering
\caption{Numerical validation of the exact volume law and exact
surface-uniform sampler.  The volume error is
$|\widehat V-V|/V$.  PIT diagnostics use $20{,}000$ exact sampler
draws.}
\label{tab:uniform-audit}
\small
\scalebox{1}{
\begin{tabular}{rrrrrrr}
\toprule
$n$ & Exact volume & MC volume & Rel. error
& Max KS & Min KS $p$ & Max $|\widehat{\rm Corr}|$\\
\midrule
2  & $4.1452\times10^{1}$  & $4.1247\times10^{1}$
   & 0.004951 & 0.009969 & 0.0373 & 0.01438\\
3  & $5.6062\times10^{2}$  & $5.6271\times10^{2}$
   & 0.003718 & 0.008661 & 0.0989 & 0.00171\\
5  & $8.8345\times10^{5}$  & $8.8032\times10^{5}$
   & 0.003541 & 0.011502 & 0.0100 & 0.01516\\
10 & $2.4868\times10^{19}$ & $2.4366\times10^{19}$
   & 0.020201 & 0.009610 & 0.0494 & 0.01997\\
\bottomrule
\end{tabular}}
\end{table}

The relative volume error is below $0.5\%$ for $n=2,3,5$ and about
$2.0\%$ for $n=10$.  Across all dimensions,
$\max D_{\rm KS}=0.011502$ and
$\max|\widehat{\operatorname{Corr}}|=0.01997$, supporting the
implemented sequential sampler.

We also validated the inverse normal-fiber coordinate map directly.
For each $n\in\{2,3,5,10\}$, $20{,}000$ ambient Gaussian observations
were inverted to $(T,\theta)$, checked for admissibility, and
reconstructed through $X(T,\theta)$.  The admissible fraction was
one in every dimension.  The maximum reconstruction errors were
$1.33\times10^{-15}$, $9.99\times10^{-16}$,
$4.33\times10^{-15}$, and $1.39\times10^{-13}$ for
$n=2,3,5,10$, respectively; the corresponding $99.9$th-percentile
errors were $6.66\times10^{-16}$, $7.22\times10^{-16}$,
$3.28\times10^{-15}$, and $1.32\times10^{-13}$.

\subsection{Projected-density validation}
\label{app:density-validation}

For the representative isotropic ambient model
$Z\sim N_{n+1}(0.75e_1,0.8^2I_{n+1})$, the projected angular density
was evaluated by Gauss-Legendre quadrature of order 96.

\paragraph{Normalization and coordinate reconstruction.}
The normalization check uses the exact Riemannian-uniform angular law
as an importance proposal.  Estimates are interpreted only when the
importance effective sample size (ESS) exceeds 50. The normalization estimates, effective sample sizes, and
normal-fiber reconstruction errors are reported in
Table~\ref{tab:density-normalization}.

\begin{table*}[ht]
\centering
\caption{Projected-density normalization and normal-fiber
reconstruction.  A dash denotes an intentionally unreported
importance estimate because ESS $<50$.}
\label{tab:density-normalization}
\small
\scalebox{1}{
\begin{tabular}{rrrrrr}
\toprule
$n$ & Integral estimate & MC SE & ESS
& Admissible fraction & Max reconstruction error\\
\midrule
2  & 1.015900 & 0.027459 & 812.61 & 1.000 & $1.33\times10^{-15}$\\
3  & 0.914758 & 0.059378 & 212.16 & 1.000 & $9.99\times10^{-16}$\\
5  & -       & -       &   4.85 & 1.000 & $4.33\times10^{-15}$\\
10 & -       & -       &   1.00 & 1.000 & $1.39\times10^{-13}$\\
\bottomrule
\end{tabular}}
\end{table*}

For $n=2$ and $n=3$, the normalization estimates are
$1.0159$ (SE $0.02746$) and $0.9148$ (SE $0.05938$),
respectively.  For $n=5$ and $n=10$, the Riemannian-uniform
importance proposal degenerates, so the corresponding absolute
normalization estimates are not interpreted.

\paragraph{Analytic density versus direct projected Monte Carlo.}
We compare expectations computed from the analytic density with those
from $20{,}000$ direct ambient-Gaussian draws followed by projection.
The analytic expectations use $2{,}000$ exact surface-uniform
importance proposals.  Probes consist of first-harmonic sine and
cosine terms and pairwise $\cos(\theta_i-\theta_j)$ terms. Table~\ref{tab:density-mc} compares the resulting analytic
expectations with their direct projected-Monte-Carlo counterparts.

\begin{table}[ht]
\centering
\caption{Agreement between analytic projected-density expectations
and direct projected Monte Carlo.  Probe errors are reported only
when ESS is at least 50.}
\label{tab:density-mc}
\small
\scalebox{1}{
\begin{tabular}{rrrrrr}
\toprule
$n$ & No. probes & ESS & Mean abs. error
& Median abs. error & Max abs. error\\
\midrule
2  &  5 & 862.51 & 0.018258 & 0.015768 & 0.033529\\
3  &  9 & 210.57 & 0.042130 & 0.039953 & 0.093125\\
5  & 20 &   2.70 & - & - & -\\
10 & 65 &   1.00 & - & - & -\\
\bottomrule
\end{tabular}}
\end{table}

The low-dimensional analytic and direct Monte Carlo calculations agree
closely.  For $n=5,10$, the importance proposal is too concentrated
to support a stable probe comparison; the direct coordinate reconstruction validation in
Table~\ref{tab:density-normalization} is therefore used instead.

\subsection{Restricted parameter-recovery study}
\label{app:recovery-audit}

We consider the identifiable restricted family
$Z\sim N_{n+1}(\alpha e_1,\sigma^2I_{n+1})$ with
$\alpha_0=0.75$ and $\sigma_0=0.8$.  For $n\in\{2,3\}$ and
$N\in\{100,250,500,1000\}$, 20 independent data sets are generated
and the projected likelihood is maximized using quadrature order 64.
All optimization runs terminate successfully. The resulting bias and RMSE values are reported in
Table~\ref{tab:mle-recovery}.

\begin{table}[ht]
\centering
\caption{Bias and RMSE for the restricted projected-Gaussian MLE over
20 repetitions.}
\label{tab:mle-recovery}
\small
\begin{tabular}{rrrrrr}
\toprule
$n$ & $N$ & Bias$(\widehat\alpha)$ & RMSE$(\widehat\alpha)$
& Bias$(\widehat\sigma)$ & RMSE$(\widehat\sigma)$\\
\midrule
2 &  100 &  0.007873 & 0.077746 & 0.013484 & 0.053995\\
2 &  250 & -0.005785 & 0.051776 & 0.002549 & 0.032500\\
2 &  500 & -0.002035 & 0.043612 & 0.008595 & 0.028170\\
2 & 1000 &  0.001683 & 0.026709 & 0.005016 & 0.020283\\
\addlinespace
3 &  100 &  0.030149 & 0.069967 & 0.004373 & 0.026065\\
3 &  250 &  0.018348 & 0.054986 & 0.010510 & 0.040358\\
3 &  500 & -0.008061 & 0.027031 & 0.011009 & 0.046859\\
3 & 1000 &  0.001929 & 0.030168 & 0.015177 & 0.041329\\
\bottomrule
\end{tabular}
\end{table}

For $n=2$, RMSE decreases clearly with sample size for both parameters.
For $n=3$, location recovery improves overall, whereas the scale RMSE
is non-monotone.  We therefore use the latter only as a numerical
diagnostic and do not infer consistency from the finite simulation.

\subsection{Intrinsic-energy solver validation}
\label{app:energy-audit}

The discrete energy solver is assessed using two independently
generated endpoint pairs for each $n\in\{2,3,5\}$.  The analytic
gradient agrees with centered finite differences to relative error
$8.76\times10^{-11}$.  Optimization uses coarse-to-fine continuation
over $M\in\{4,8,16,32\}$, multiple coarse initializations,
analytic-gradient L-BFGS-B, and final BFGS polishing when needed.
For each coordinate, the two terminal lifts nearest to the initial
angle are retained, giving $2^n$ candidate winding classes.

The prespecified convergence criteria are
$$
\frac{|E_{32}-E_{16}|}{E_{16}}<0.01,
\quad
\|\nabla E_{32}\|_\infty<10^{-4}.
$$

Table~\ref{tab:energy-final} reports the mesh-refinement error,
final gradient residual, and selected winding class for each of the
six endpoint pairs.

\begin{table*}[ht]
\centering
\caption{Final intrinsic-energy convergence results.  Every case
satisfies the $1\%$ mesh-refinement and $10^{-4}$ first-order
stationarity criteria.}
\label{tab:energy-final}
\small
\scalebox{1}{
\begin{tabular}{ccccccc}
\toprule
$n$ & Pair & $E_{16}$ & $E_{32}$ & Rel. change
& $\|\nabla E_{32}\|_\infty$ & Selected winding\\
\midrule
2 & 1 & 18.652898  & 18.627892  & 0.001341
  & $8.24\times10^{-7}$ & $(0,1)$\\
2 & 2 &  7.566471  &  7.562600  & 0.000512
  & $9.17\times10^{-7}$ & $(0,0)$\\
3 & 1 &  8.882255  &  8.877530  & 0.000532
  & $1.47\times10^{-6}$ & $(0,0,0)$\\
3 & 2 &  6.559641  &  6.558792  & 0.000129
  & $6.48\times10^{-7}$ & $(-1,0,0)$\\
5 & 1 & 116.109317 & 115.284795 & 0.007101
  & $1.53\times10^{-6}$ & $(0,1,0,0,0)$\\
5 & 2 & 326.328870 & 323.765835 & 0.007854
  & $1.76\times10^{-5}$ & $(0,0,0,1,0)$\\
\bottomrule
\end{tabular}}
\end{table*}

The worst final mesh change is $0.7854\%$ and the largest final
gradient infinity norm is $1.76\times10^{-5}$.  Nonzero selected
windings occur in several cases, confirming that restricting the
optimization to principal angular representatives can miss a
lower-energy lifted path.  These results establish numerical
mesh stability and first-order stationarity for the representative
problems; they are not a theorem of global convergence. 
\begin{figure}[t]
\centering
\includegraphics[width=0.7\textwidth]{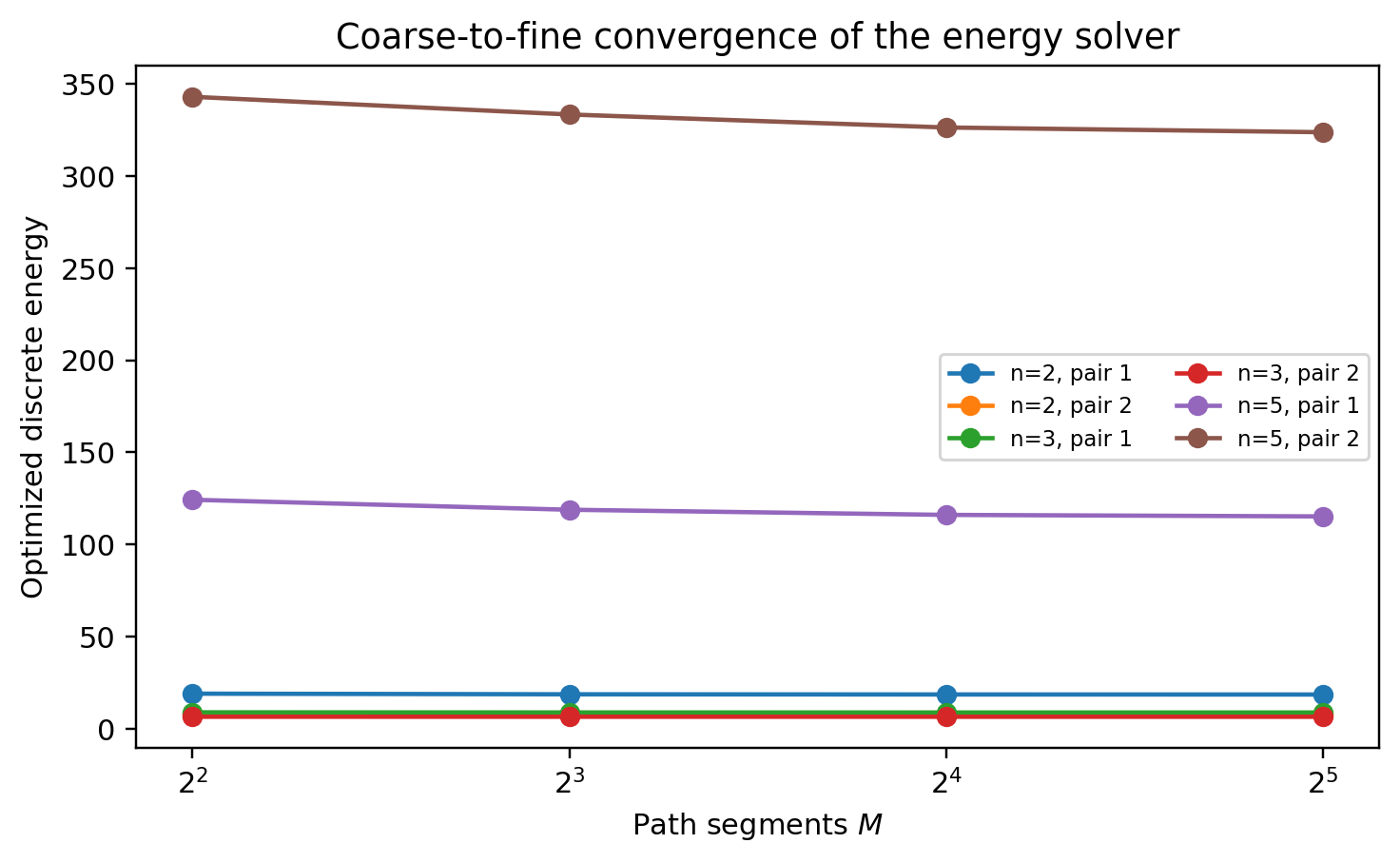}
\caption{Coarse-to-fine convergence of optimized intrinsic path energy for two endpoint pairs in $n=2,3,5$.}
\label{fig:energy}
\end{figure}

Figure~\ref{fig:energy} shows the corresponding coarse-to-fine
energy trajectories and illustrates stabilization as the mesh is
refined from $M=4$ to $M=32$.

\paragraph{Computational note.}
Exact sampling and one-dimensional projected-density evaluation remain
inexpensive over the investigated dimensions, whereas intrinsic-energy
optimization is substantially more expensive because it combines path
optimization with a winding search.  Exact wall-clock values are
hardware dependent and are recorded in the reproducibility outputs
rather than treated as fixed scientific quantities here.


\section{Additional real-data results and diagnostics}
\label{app:realdata}

\subsection{Model-specific numerical diagnostics}
\label{app:realdata-diagnostics}

The real-data comparison combines likelihood-based, pseudo-likelihood,
Bayesian, and simulation-based estimators, so diagnostics are reported
model by model rather than through a common optimizer criterion.

\paragraph{Wrapped-normal truncation.}
For the Franka analysis, re-evaluating M1 with the larger wrapping
lattice gives
$
\mathrm{NLL}_{K=2}=\mathrm{NLL}_{K=3}=-10.549207,
$
showing that the finite lattice is numerically adequate for this data.

\paragraph{M2 pseudo-likelihood stability.}
Five legitimate starting values give the same Franka training
pseudo-negative-log-likelihood,
$-32043.669598,
$
with total spread $1.85\times10^{-8}$, corresponding to relative
spread $5.76\times10^{-13}$.  The held-out Energy, MMD, and SWD values
are also stable across these starts.

\paragraph{M4 quadrature convergence.}
For Franka, Gauss-Legendre orders 96, 128, 192, and 256 all give
$
\mathrm{NLL}=3.918539
$
to six decimals; the difference between orders 192 and 256 is at
approximately machine precision.

\paragraph{Fixed-geometry sensitivity.}
The primary geometry margin $\delta=0.05$ was fixed before official
test evaluation. The primary geometry margin $\delta=0.05$ was fixed before official
test evaluation. Re-fitting M4 over
$\delta\in\{0.02,0.05,0.10,0.20\}$ gives the held-out discrepancies
reported in Table~\ref{tab:delta-sensitivity}.

\begin{table}[ht]
\centering
\caption{Franka fixed-geometry sensitivity for M4.  Lower is better.}
\label{tab:delta-sensitivity}
\small
\begin{tabular}{cccc}
\toprule
$\delta$ & Energy & $\MMD^2$ & SWD\\
\midrule
0.02 & 0.03713 & 0.05927 & 0.04614\\
0.05 & 0.03913 & 0.06202 & 0.04914\\
0.10 & 0.03551 & 0.05785 & 0.04439\\
0.20 & 0.03910 & 0.05981 & 0.04845\\
\bottomrule
\end{tabular}
\end{table}

Because each sensitivity setting is refitted and evaluated by Monte
Carlo generation, the $\delta=0.05$ row need not reproduce the
primary-table Monte Carlo mean exactly.  The prespecified
$\delta=0.05$ analysis is retained; the test data are not used to
select $\delta$.

For M3 and M5, posterior-predictive evaluation uses retained MCMC draws
after burn-in and thinning.  Trace, acceptance, seed, and run-setting
diagnostics are retained in the accompanying reproducibility outputs.

\paragraph{Reproducibility settings.}
For Franka, the analysis uses $d=4$, joints \texttt{pos1-pos4},
frame stride 8, validation trial \texttt{007}, five official test
trajectories, 10 generated-sample evaluation replicates per trajectory,
and 300 SWD directions.  M4 uses
$R=(1,1.05,2.1525,4.412625)$, covariance floor $10^{-3}$, and three
validation-selected restarts.  For ALOHA Coffee, the first four
left-arm revolute joints are retained in physical serial order, every
fifth frame is used, and the split is 35 fitting, 7 validation, and
8 untouched test episodes.  Cysteine is analyzed separately within
CYS, CSO, CSD, and OCS using the 70/15/15 split described in the main
paper; the reported macro-average gives equal weight to the four states.

\subsection{Paired held-out robot comparisons}
\label{app:paired-robot}

For each held-out Franka trajectory and ALOHA Coffee episode, we
compute M4-minus-competitor differences.  Negative values favor M4.
Table~\ref{tab:paired_m3_m5} focuses on M3 and M5, the projected-normal
comparators that most directly separate the effect of the recursive
geometry from projected-normal flexibility.  These are descriptive
paired summaries; no independence assumption across frames is used. The aggregate held-out discrepancy patterns for all six models are
shown in Figure~\ref{fig:robot-realdata-supp}.

\begin{table*}[ht]
\centering
\caption{Paired held-out M4-minus-competitor discrepancies for the
projected-normal comparators.  ``Wins'' is the number of held-out
trajectories or episodes for which the paired difference is negative.}
\label{tab:paired_m3_m5}
\small
\begin{tabular}{llcccc}
\toprule
Dataset & Comparator & Metric & Mean difference & SE & M4 wins\\
\midrule
Franka & M3 & Energy   & -0.044230 & 0.003303 & 5/5\\
Franka & M3 & $\MMD^2$ & -0.022089 & 0.002625 & 5/5\\
Franka & M3 & SWD      & -0.045476 & 0.002793 & 5/5\\
Franka & M5 & Energy   & -0.010965 & 0.001568 & 5/5\\
Franka & M5 & $\MMD^2$ & -0.060439 & 0.002307 & 5/5\\
Franka & M5 & SWD      &  0.003608 & 0.001395 & 0/5\\
\midrule
ALOHA Coffee & M3 & Energy   & -0.001792 & 0.004163 & 6/8\\
ALOHA Coffee & M3 & $\MMD^2$ & -0.003271 & 0.003500 & 6/8\\
ALOHA Coffee & M3 & SWD      &  0.002676 & 0.001288 & 1/8\\
ALOHA Coffee & M5 & Energy   & -0.007087 & 0.003812 & 6/8\\
ALOHA Coffee & M5 & $\MMD^2$ & -0.007230 & 0.003740 & 6/8\\
ALOHA Coffee & M5 & SWD      & -0.004498 & 0.001428 & 7/8\\
\bottomrule
\end{tabular}
\end{table*}

For Franka, M4 has lower Energy and $\MMD^2$ than both projected-normal
comparators on all five trajectories. For ALOHA Coffee, the paired
comparisons favor M4 over M3 for Energy and $\MMD^2$ on 6 of 8
episodes, while M3 is favored under SWD on 7 of 8 episodes.

\begin{figure}[ht]
\centering
\includegraphics[width=0.45\textwidth]
{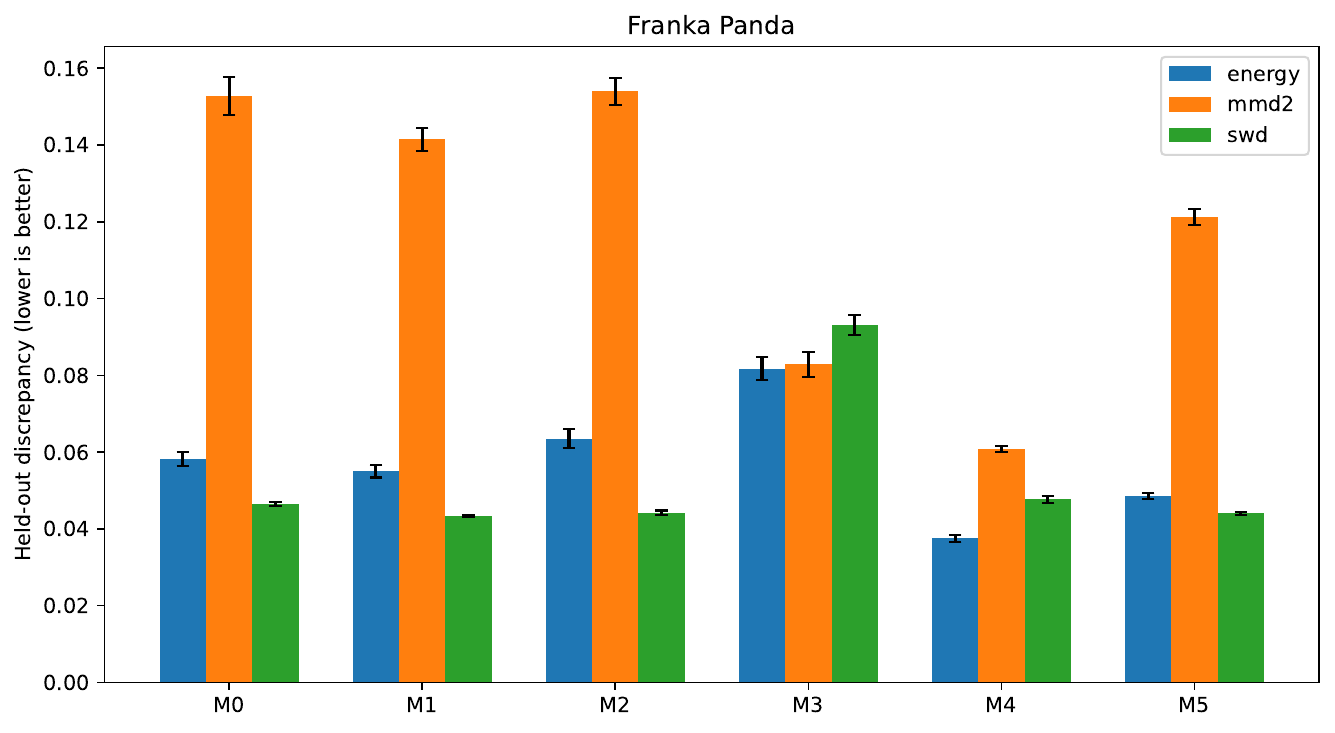}
\hfill
\includegraphics[width=0.45\textwidth]
{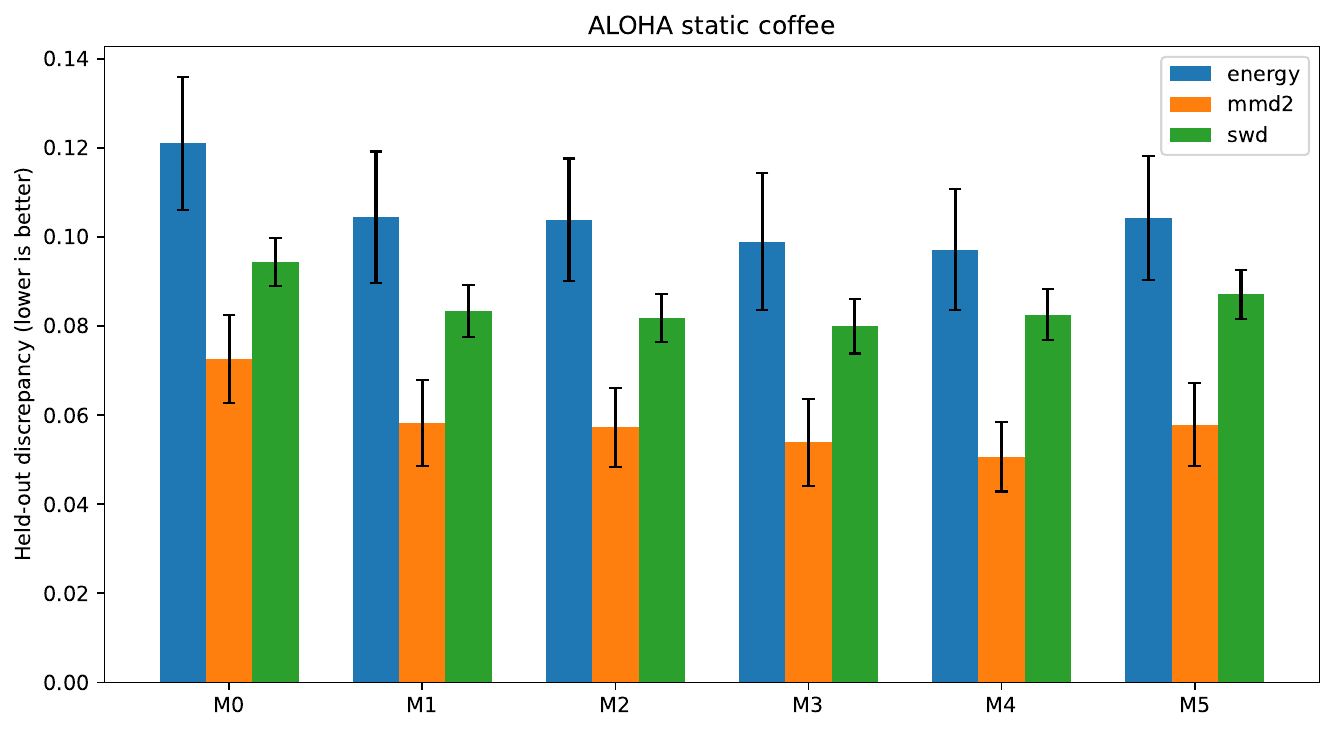}
\caption{Held-out discrepancy summaries for the two robot applications.
\textbf{Left:} Franka Panda lift-hold.  \textbf{Right:} ALOHA Coffee.  Error bars are
standard errors across held-out trajectories or episodes.}
\label{fig:robot-realdata-supp}
\end{figure}
\subsection{Cysteine state-wise results}
\label{app:cysteine-statewise}

The four oxidation states are modeled separately.
Table~\ref{tab:cysteine-statewise-m0-m5} gives the complete
state-wise comparison across M0-M5.  Unlike the two robot
applications, M3 attains the lowest held-out Energy, $\MMD^2$, and
SWD in every oxidation state.

\begin{table}[ht]
\centering
\caption{Cysteine state-wise held-out discrepancies.  Lower is better
within each state and metric.}
\label{tab:cysteine-statewise-m0-m5}
\small
\begin{tabular}{llccc}
\toprule
State & Model & Energy & $\MMD^2$ & SWD\\
\midrule
CYS & M0 & 0.065661 & 0.024435 & 0.137595\\
    & M1 & 0.075083 & 0.026505 & 0.146259\\
    & M2 & 0.063582 & 0.022692 & 0.133466\\
    & M3 & \textbf{0.019379} & \textbf{0.005042} & \textbf{0.070304}\\
    & M4 & 0.026082 & 0.006804 & 0.085384\\
    & M5 & 0.072415 & 0.025394 & 0.144145\\
\midrule
CSO & M0 & 0.094395 & 0.026448 & 0.168758\\
    & M1 & 0.107180 & 0.028922 & 0.178106\\
    & M2 & 0.096150 & 0.024486 & 0.167261\\
    & M3 & \textbf{0.059222} & \textbf{0.014547} & \textbf{0.125811}\\
    & M4 & 0.070888 & 0.017537 & 0.141233\\
    & M5 & 0.107702 & 0.028508 & 0.176967\\
\midrule
CSD & M0 & 0.153193 & 0.034378 & 0.222562\\
    & M1 & 0.148546 & 0.033073 & 0.219538\\
    & M2 & 0.134730 & 0.028255 & 0.206804\\
    & M3 & \textbf{0.104712} & \textbf{0.022196} & \textbf{0.184455}\\
    & M4 & 0.132257 & 0.028435 & 0.204796\\
    & M5 & 0.129306 & 0.027063 & 0.203227\\
\midrule
OCS & M0 & 0.341298 & 0.103300 & 0.306123\\
    & M1 & 0.297330 & 0.089368 & 0.290278\\
    & M2 & 0.264483 & 0.077866 & 0.268735\\
    & M3 & \textbf{0.158564} & \textbf{0.041216} & \textbf{0.203185}\\
    & M4 & 0.295514 & 0.087216 & 0.275139\\
    & M5 & 0.227555 & 0.068042 & 0.251460\\
\bottomrule
\end{tabular}
\end{table}

The equal-weight macro-averages reported in the main text therefore
summarize a state-wise pattern that consistently favors M3 rather than
being driven by a single oxidation state.

\subsection{Sensitivity to angular-coordinate ordering}
\label{app:ordering-sensitivity}

Because the recursively curved embedding is not permutation invariant,
we refitted M4 under all $4!=24$ coordinate permutations for Franka and
ALOHA Coffee, using training and validation data only.  The physical
proximal-to-distal order remains the prespecified primary order and is
not selected from this sensitivity analysis.  For cysteine, the
conventional $(\phi,\psi)$ order is compared with $(\psi,\phi)$.
Table~\ref{tab:ordering-sensitivity} summarizes the validation-only
performance of the primary ordering relative to the best ordering
considered in each application.

\begin{table}[ht]
\centering
\caption{Validation-only sensitivity of M4 to angular-coordinate
ordering.  ``Gap'' is
$100\{\MMD_{\rm primary}^2/\MMD_{\rm best}^2-1\}$.}
\label{tab:ordering-sensitivity}
\small
\scalebox{1}{\begin{tabular}{lcccc}
\toprule
Dataset & Primary val. $\MMD^2$ & Best val. $\MMD^2$
& Rank & Gap (\%)\\
\midrule
Franka lift-hold & 0.055224 & 0.049041 & 5/24 & 12.61\\
ALOHA Coffee      & 0.008308 & 0.008308 & 1/24 & 0.00\\
Cysteine macro    & 0.015284 & 0.008740 & 2/2  & 74.86\\
\bottomrule
\end{tabular}}
\end{table}

The physical ALOHA order is the validation optimum, while the Franka
physical order ranks fifth among the $24$ permutations.  For cysteine,
reversing $(\phi,\psi)$ reduces the macro validation $\MMD^2$ from
$0.015284$ to $0.008740$.  Figure~\ref{fig:ordering-sensitivity}
shows the corresponding ordering-sensitivity profiles.  These results
support treating coordinate order as externally specified when the
application provides a natural sequence, and otherwise as a modeling
choice to be assessed without access to the test data.

\begin{figure}[t]
\centering
\begin{minipage}[t]{0.45\textwidth}
\centering
\includegraphics[width=\linewidth]
{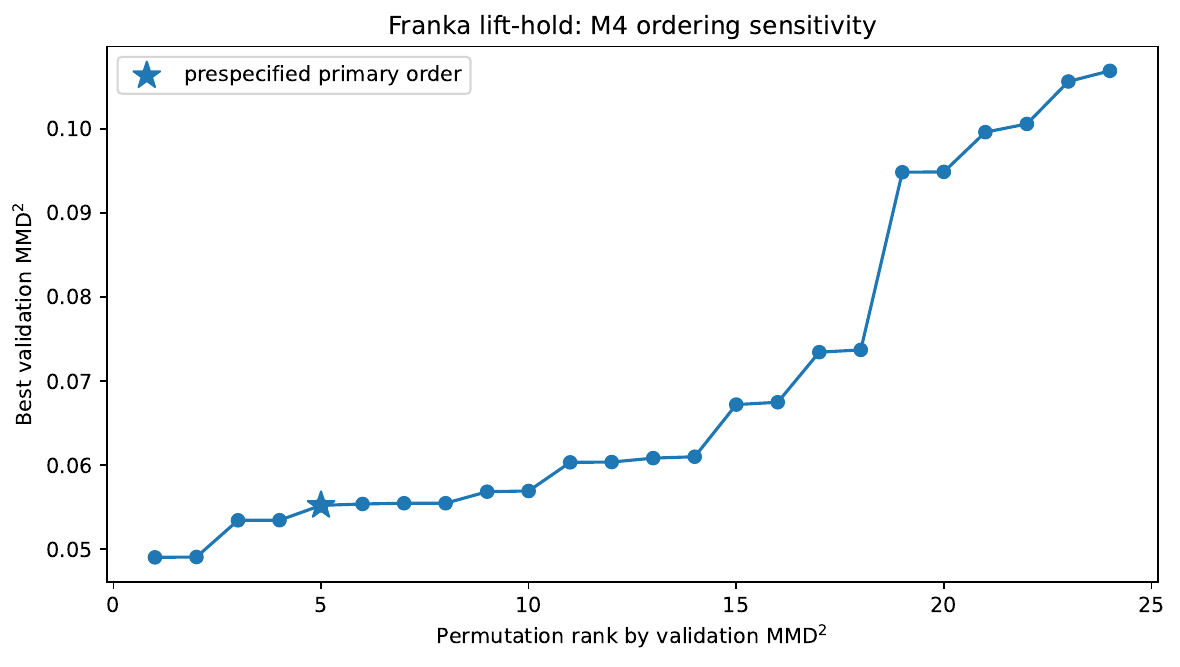}

 (a) Franka lift-hold
\end{minipage}
\begin{minipage}[t]{0.45\textwidth}
\centering
\includegraphics[width=\linewidth]
{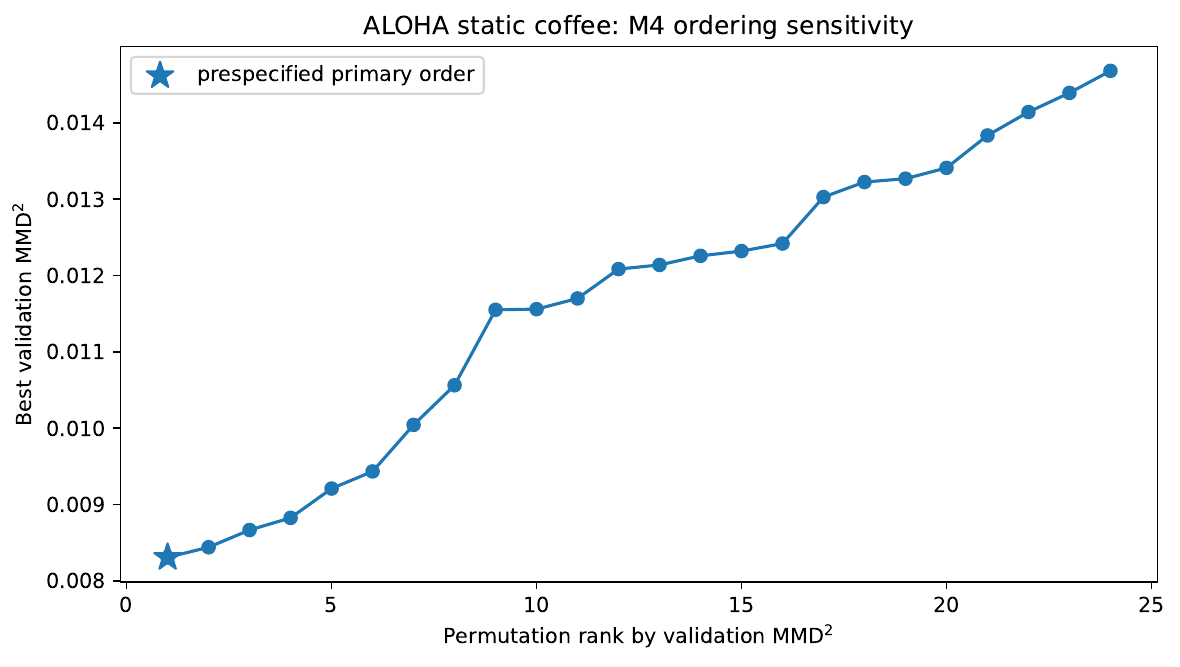}

 (b) ALOHA Coffee
\end{minipage}
\hfill
\begin{minipage}[t]{0.45\textwidth}
\centering
\includegraphics[width=\linewidth]
{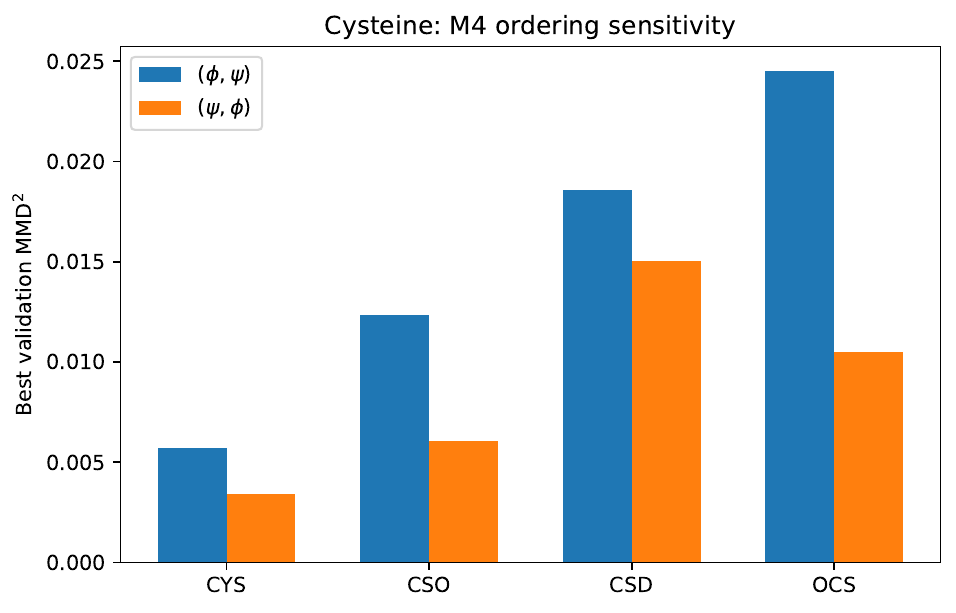}

 (c) Cysteine
\end{minipage}

\caption{Validation-only sensitivity of M4 to angular-coordinate
ordering.  For Franka and ALOHA Coffee, all $4!=24$ permutations are
evaluated using validation data only; the physical proximal-to-distal
ordering remains the prespecified primary analysis.  For cysteine, the
conventional $(\phi,\psi)$ and reversed $(\psi,\phi)$ orderings are
compared.  Lower validation $\MMD^2$ indicates a better fit.  No
held-out test data are used to select the ordering.}
\label{fig:ordering-sensitivity}
\end{figure}
\end{document}